\documentclass[preprint]{elsarticle}
\usepackage{graphicx,mathrsfs,amsmath,amsfonts,amsthm,amssymb}
\usepackage{booktabs}
\usepackage[table,xcdraw]{xcolor}
\usepackage{rotating,tabularx,siunitx}
\usepackage{multirow}
\usepackage{nicefrac}
\usepackage{mathtools}
\usepackage{subcaption}
\usepackage{tikz}
\usepackage{url}

\definecolor{cellcolor}{rgb}{0.94, 0.92, 0.84}

\newtheorem{thm}{Theorem}[section]
\newtheorem{prop}[thm]{Proposition}
\newtheorem{defi}[thm]{Definition}
\newtheorem{lemma}[thm]{Lemma}

\newtheorem{rem}[thm]{Remark}

\newcommand{\NN}{\mathbb{N}}
\newcommand{\AC}{\mathscr{A}}
\newcommand{\ST}{\mathscr{S}}
\newcommand{\CC}{\mathcal{C}}
\newcommand{\NP}{\mathcal{NP}}
\newcommand{\RR}{\mathbb{R}}
\newcommand{\ZZ}{\mathbb{Z}}

\journal{}

\begin{document}

\begin{frontmatter}

\title{Solving the List Coloring Problem through a Branch-and-Price algorithm\tnoteref{grant}}

\author[a,b]{Mauro Lucci}
\author[a,b]{Graciela Nasini}
\author[a,b]{Daniel Sever\'in\fnref{correspon}}

\address[a]{Depto. de Matem\'atica (FCEIA), Universidad Nacional de Rosario, Argentina}

\address[b]{CONICET, Argentina}

\tnotetext[grant]{Partially supported by grants PICT-2020-03032 (ANPCyT) and PIP-1900 (CONICET).\\
\emph{E-mail addresses}: \texttt{mlucci@fceia.unr.edu.ar} (M. Lucci),
\texttt{nasini@fceia.unr.edu.ar} (G. Nasini),
\texttt{daniel@fceia.unr.edu.ar} (D. Sever\'in).}
\fntext[correspon]{Corresponding author at Departamento de Matem\'atica (FCEIA), UNR, Pellegrini 250, Rosario, Argentina.}

\begin{abstract}
In this work, we present 
a branch-and-price algorithm to solve the weighted version of the List Coloring Problem, based on a vertex cover formulation by stable sets. This problem   is interesting for its applications and  also for the many other problems that it generalizes, including the well-known Graph Coloring Problem.
With the introduction of the concept of indistinguishable colors, some theoretical results are presented which are later incorporated into the algorithm.
We propose two branching strategies based on others for the Graph Coloring Problem, the first is an adaptation of the one used by Mehrotra and Trick in their pioneering branch-and-price algorithm 
and the other is inspired by the one used by M\'endez-D\'iaz and Zabala in their branch-and-cut algorithm. 
The rich structure of this problem makes both branching strategies robust.
Extended computation experimentation on a wide variety of instances shows the effectiveness of this approach and evidences the different behaviors that the algorithm can have according to the structure of each type of instance.
\end{abstract}
\begin{keyword}
List Coloring, Branch and Price, Weighted Problem. 
\MSC[2010] 90C57 \sep 05C15 
\end{keyword}
\end{frontmatter}


\section{Introduction} \label{intro}

The \emph{Graph Coloring Problem} (GCP) models a wide range of planning problems such as timetabling, scheduling, electronic bandwidth allocation, sequencing, and register allocation.
Because of this, many efforts have been made to develop high-performance tools to solve this $\NP$-hard problem.

Given an undirected simple graph $G = (V,E)$  and a set of colors $\CC$, a $\CC$-\emph{coloring} (or just a \emph{coloring}) of $G$ is a function $f : V \rightarrow \CC$ such that $f(u) \neq f(v)$ for every edge $(u,v)$ of $G$.
Given a coloring $f$ of $G$ and $j\in \CC$, $f^{-1}(j) \doteq \{v \in V: f(v) = j\}$ is the \emph{color class of} $j$, i.e.~ the subset of vertices colored by $j$. 
The \emph{active colors in} $f$, denoted by $\AC(f)$, is the set of colors assigned to some vertex, i.e.~$\AC(f) \doteq \{j \in \CC : f^{-1}(j) \neq \emptyset\}$.
Given $k\in \ZZ_+$, a coloring $f$ is a $k$-coloring if $|\AC(f)|=k$. 
The GCP consists of finding the minimum $k$ such that $G$ admits a $k$-coloring. 
This minimum denoted by $\chi(G)$ is
called \emph{the chromatic number of} $G$.

In the context of exact approaches, best known algorithms are based on integer linear programming (ILP) models and solved by different
branch-and-bound schemes such as branch-and-cut (B\&C) \cite{mendezdiaz2006} and branch-and-price (B\&P) \cite{trick1996}.
Novel features added to these classical approaches have consistently increased the size of the solved instances of coloring problems,
e.g.~\cite{
Furini2017, malaguti2012}. 


\medskip

In some applications, it becomes necessary to consider some additional restrictions, giving rise to variants of colorings \cite{tuza1997,kubale2004}.

For instance, $\mu$-colorings  \cite{bonomo2011} 
arise in applications where resources are ordered by quality or size, and users have a particular minimum quality
or size requirement.
As an example, the problem of assigning classrooms (with different capacities) to courses (with different number of students). 


Formally, given a graph $G=(V,E)$, a set of colors $\CC=\{1,\ldots, n\}$ with $n\in \NN$, and a vector $\mu \in \CC^V$, a \emph{$\mu$-coloring} of $G$ is a coloring $f$ of $G$ such that $f(v) \leq \mu(v)$, for all $v \in V$. 
The $\mu$-Coloring Problem consists of finding a $\mu$-coloring $f$ of $G$ with minimum cardinality of $\AC(f)$. 

In the context of scheduling and very large-scale integration, some problems can be modeled as extending a given \emph{partial coloring} in a graph, i.e.~finding a coloring of a graph where some vertices have a preassigned color \cite{biro1992}.
Formally, in the Precoloring Extension Problem, the input is a graph $G=(V,E)$, a set of colors $\CC$, a subset $W \subset V$, and a coloring $f'$ of the subgraph of $G$ induced by $W$. The goal is to find a coloring $f$ of $G$ such that $f(v) = f'(v)$ for all $v \in W$, with minimum cardinality of $\AC(f)$. Precoloring Extension Problem also models theoretical problems such as the Quasigroup Completion 
\cite{quasigroup}. 
Moreover, solving GCP, a well-known and highly-effective preprocessing stage consists of coloring a maximal clique of the graph, transforming the GCP instance into an instance of Precoloring Extension Problem.

\medskip

In this paper we address a coloring problem which generalizes all those mentioned previously: the List Coloring Problem (LCP).
In the LCP each vertex $v$ of $G$ has a non-empty preassigned \emph{list of colors} $L(v)\subset \mathcal C$. Instances are denoted by $(G,L)$.
A \emph{list coloring} of $(G,L)$ is a 
coloring $f$ of $G$ with the additional condition that $f(v) \in L(v)$ for all $v\in V$. W.l.o.g.~we can consider
$\CC= \bigcup_{v\in V} L(v)$. 
As usual, LCP consists of finding a list coloring of $(G,L)$ with minimum cardinality of $\AC(f)$.

Clearly, GCP and $\mu$-Coloring Problem are equivalent to LCP on instances where, for all $v\in V$,  $L(v)=\mathcal C$ and $L(v)=\{1,\ldots, \mu(v)\}$, respectively. In addition, the Precoloring Extension Problem is equivalent to LCP on instances where $L(v) = \{f'(v)\}$ for all $v \in W$ and 
$L(v) = \CC$ for all $v \in V \setminus W$.

\medskip

There are practical situations where we need to consider weights (or costs) on colors
leading to the \emph{Weighted List Coloring Problem} (WLCP).
Given a weight vector $w\in \ZZ_+^{\CC}$ and a coloring $f$ of $G$, its weight $w(f)$ is defined as $w(f) \doteq \sum_{j \in \AC(f)} w_j$.
The goal of the WLCP is to find the list coloring $f$ with minimum weight. 
Formally:

\begin{center}
\fbox{\parbox{0.9\textwidth}{
\textsc{Weighted List Coloring Problem} (WLCP)\\[5pt]
\emph{Instance}: a graph $G=(V,E)$, a set of finite lists $L=\{ L(v):v\in V\}$,
and $w\in \ZZ^\CC_+$.\\[5pt]
\emph{Objective}: obtain a list coloring $f$ of $(G,L)$ such that $w(f)$ is minimum or answer that such a coloring does not exist.}}\\ 
\end{center}

Instances of WLCP will be denoted by $(G,L,w)$. 
Clearly, instances of the LCP can be seen as particular instances of the WLCP with $w_c=1$ for all $c\in \mathcal C$. 
It is interesting to point out that the \emph{weighted} version of GCP is equivalent to the GCP. 
Indeed, an optimal solution can be obtained from a $\chi(G)$-coloring of $G$ where the active colors are
those with the first $\chi(G)$ lowest weights.
It is easy to see that this is not the case for general instances of WLCP.

\medskip


\medskip

Regarding  computational complexity, since GCP is $\NP$-hard in general graphs, LCP (as well as WLCP) also is. However, some known results seem to reveal that LCP is  \emph{harder} than GCP.

In particular, feasibility is not easy to handle. 
The well-known example presented in Figure \ref{fig:infeas} corresponds to an infeasible instance of the LCP where all the vertices have $\chi(G)$ elements in their color lists.
\begin{figure}[h]
	\centering
		\includegraphics[scale=0.2]{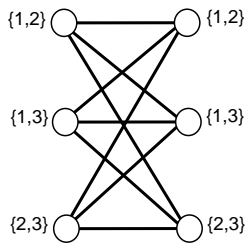}
		\caption{An infeasible instance of LCP. Color lists are displayed between braces.} \label{fig:infeas}
\end{figure}

Deciding  whether $(G,L)$ has a list coloring is an $\NP$-complete problem, even when restricted to instances where the size of the lists is at most 3 and $G$ is a cograph or complete bipartite \cite{golovach2017}. 
In contrast, GCP can be solved in linear time for cographs \cite{CHVATAL} and bipartite graphs. 
Besides, for  graphs $G$ with bounded treewidth $t$, there is a dynamic programming algorithm that finds its chromatic number in polynomial time, while asking if $G$ has a list coloring is W[1]-Hard \cite{Computcomplex}.
Despite this, there is an effort to design good algorithms that find a list coloring, whenever it exists. As far as we know, the best one is $O(2^n)n^{O(1)}$ for a graph of $n$ vertices and lists of arbitrary size \cite{husfeldt2009} .




\medskip 

A practical application of WLCP can be found in the design of workdays of drivers in a public transport company \cite{ejs2018}.
In addition, besides coloring problems, the WLCP also generalizes other optimization problems, e.g. the well-known Minimum Weighted Set Cover Problem.
In the next section, we will analyze the relationship between these two problems.

In conclusion, the development of a tool to solve general instances of the WLCP
would allow addressing a wide range of combinatorial optimization problems and  applications.

\medskip 

In this work, we propose an exact algorithm for the WLCP based on an ILP formulation. This research extends some of the ideas presented by the authors in the conference paper \citet{entcs2019}.

In Section \ref{Sec:Preliminaries} necessary notations and definitions are presented, as well as some theoretical results on which our algorithm is based.  The ILP formulation with an exponential number of variables is presented in Section \ref{ilpsc}, based on the characterization of list colorings as covers of  vertices by stable sets.

In Sections \ref{Sec:Solving:LR}, \ref{Sec:branchingstrategies}, and  
\ref{Sec:branchingvariable}, the main components of our B\&P algorithm are presented: pricing routing, branching strategies, and variable selection rules.

The lack of benchmark instances for the LCP and the variety of relevant parameters that define them led us to design 5 sets of instances that cover a wide variety of structures presented in Section \ref{Sec:Instances}. Extended computational experiences are evaluated on these sets of instances in Section \ref{Sec:computationalexperiences}.

Finally, in Section \ref{Sec:Conclusions} we present some conclusion and pending tasks.


\section{
Preliminaries}\label{Sec:Preliminaries}

Given a graph $G=(V,E)$ and $v\in V$, $N_G(v)$ denotes the set of vertices adjacent to $v$, and $N_G[v] = N_G(v) \cup \{v\}$.
Given $U\subset V$, $G[U]$ denotes the subgraph of $G$ induced by $U$ and $G - U=G[V\setminus U]$. When $U=\{u\}$ we simply denote $G-u$.
A subset of vertices is a \emph{stable set} (resp. \emph{clique}) of $G$ if no (resp. all) pair of vertices in $S$ are adjacent to each other. 

Given an instance $(G,L,w)$ of WLCP, for each $j \in \CC \doteq \bigcup_{v\in V} L(v)$, we define $V_j \doteq \{ v \in V : j \in L(v) \}$ and $G_j \doteq G[V_j]$, the \emph{color-vertex set} and the
\emph{color graph} (associated to $j$), respectively. Additionally, 
$\ST_j$ is the set of non-empty stable sets of $G_j$.

Observe that instances of WLCP can also be described by giving the color-vertex sets instead of the lists of colors. 
Indeed, a 3-tuple $(G,\mathcal V, w)$ with $G=(V,E)$, $\mathcal V=\{V_j\subset V: j\in \CC\}$, and $w\in \ZZ_+^{\CC}$  corresponds to the instance $(G,L,w)$ where  $L(v) = \{j \in \CC : v \in V_j \}$.
Since both representations are equivalent,  we will use one or the other interchangeably depending on the context.

\bigskip

The following result reveals that the structure of an instance $(G,L)$ of LCP depends on  color graphs  rather than $G$,  as edges with endpoints in different color-vertex sets can be added to (or removed from) $G$ without modifying the set of list colorings.

\begin{lemma} \label{uselessedges}
Let $(G,L)$ with $G=(V,E)$ be an instance of the LCP and  
$u,v \in V$ such that $(u,v)\notin E$ and $L(u) \cap L(v) = \emptyset$.
Let $G'$ be the graph obtained by adding the edge $(u,v)$ to $G$.
Then, 
$f$ is a list coloring of $(G,L)$ if and only if $f$ is a list coloring of $(G',L)$. 
\end{lemma}
\begin{proof}
Let $f$ be a list coloring of $(G,L)$. 
Since $f(u) \in L(u)$, $f(v) \in L(v)$ and, by hypothesis, $L(u) \cap L(v) = \emptyset$, then $f(u) \neq f(v)$.
Therefore, $f$ is a list coloring of $(G',L)$.

The converse is straightforward.
\end{proof}

Based on the structure of the color graphs,  we present two families of instances that are \emph{opposite} in terms of computational complexity. We prove that the WLCP is polynomial-time solvable when all color graphs are complete graphs and remains $\NP
$-hard when color graphs are edgeless graphs.

The former can be reduced to the \emph{Minimum Weighted Perfect Matching Problem} (MWPMP) on bipartite graphs, which can be solved in polynomial time by the Hungarian algorithm (see e.g.~\cite{edmonds1972}). Indeed, we have:

\begin{prop} \label{completeispoly}
The WLCP is polynomial-time solvable on instances $(G,L,w)$ such that, for all $j \in \CC$, $G_j$ 
is a complete graph.

\end{prop}
\begin{proof}
Let $G=(V,E)$ and $(G,L,w)$ be an instance of the WLCP such that, for all $j \in \CC$, $G_j$ 
is a complete graph.
Since $(G,L,w)$ is infeasible when $|V| > |\CC|$, we assume that
$|V| \leq |\CC|$. 

Let $Z$ be a set of dummy elements such that $|Z| = |\CC| - |V| $. We construct the bipartite graph $H$ whose vertex partition is defined by the sets $V \cup Z$ and $\CC$. Moreover, each $v \in V$ is adjacent to every $j \in L(v)$ and each $z \in Z$ is adjacent to every $j \in \CC$.
In addition, edges $(v,j)$ with $v \in V$ have weight $w_j$ while edges $(z,j)$ with $z \in Z$ have weight zero.

Given a perfect matching $M$ of $H$, let $f_M : V \rightarrow \CC$ be such that $f_M(v) = j$ if $(v,j) \in M$.
By construction, $f_M$ is a list coloring of $(G,L)$ and $w(f_M) = w(M)$, where $w(M)$ is the weight of $M$.

Conversely, given a list coloring $f$ of $(G,L)$ we construct $M_f = \{ (v,f(v)) : v \in V \}$. Clearly, $M_f$ is a matching of $H$ such that $w(M_f)=w(f)$. Let $g$ be a one-to-one correspondence between $Z$ and $\CC\setminus f(V)$. To obtain a perfect matching $M$ of $H$, we add to $M_f$ arcs $(z,g(z))$, for all $z\in Z$. Clearly, $w(M)=w(f)$. 

Therefore, $(G,L,w)$ is an infeasible instance of $WLCP$ if and only if $H$ has no perfect matching. Moreover, if $M$ is a minimum weighted 
perfect matching of $H$ then $f_M$ is an optimal list coloring of $(G,L,w)$.  
\end{proof}

For the case of edgeless color graphs, we obtain a polynomial-time reduction from the Minimum Weighted Set Cover Problem (MWSCP).


Given a set $X$ and a family $\mathcal S$ of subsets of $X$ such that $\bigcup_{S \in \mathcal{S}} S = X$ a \emph{set cover of} $X$ is a subset $\mathcal S'$ of $\mathcal S$ such that $\bigcup_{S \in \mathcal{S}'} S = X$. 
In the MWSCP, the elements of $\mathcal S$ have weights, given by a vector $w\in \ZZ_+^{\mathcal S}$, and the objective is to find a set cover $\mathcal S'$ of $X$ with minimum weight $\sum_{S\in \mathcal S'} w_S$. The MWSCP is $\NP$-hard even for the unweighted case ($w=\mathbf 1$), with sets of size 3 (see the problem [SP5] of \cite{GAREY}).


We can prove:

\begin{prop} \label{emptyisnp}
The LCP is $\NP$-hard on instances $(G,L)$ where, for all $j \in \CC$, $G_j$ 
is an edgeless graph.
\end{prop}

\begin{proof}
Let $(X, \mathcal S, \mathbf 1)$ be an instance of MWSCP. We construct the LCP instance  $(G,L)$ where $G=(X,\emptyset)$ and $L(v) = \{ S \in \mathcal S : v \in S\}$ for all $v\in X$. 

Given a list coloring  $f$ of $(G,L)$, clearly $\mathcal A (f)$
is a set cover of $X$.
Conversely, given a set cover  $\mathcal S'$  of $X$, there exists a function $\tilde f: X \to \mathcal S'$ such that,
for all $v\in X$,  
$v\in \tilde f(v)$. Clearly, $\tilde f$ is a list coloring of $(G,L)$.
Moreover, if $f$ is an optimal list coloring of $(G,L)$ we have
$|\mathcal S'|\geq |\mathcal A(\tilde f)| \geq |\mathcal A(f)|$.
Then $\mathcal A(f)$ is an optimal set cover of $X$.

Observe that we also prove that  $(X, \mathcal S, \mathbf 1)$ is infeasible if and only if $(G,L)$ also is. 

\end{proof}

\medskip

It is known that given a coloring $f$ of a graph $G$ and a color $j\in \mathcal A(f)$, $f^{-1}(j)$ is a stable set of $G$. This condition allows defining colorings as partitions of the set of vertices into stable sets.
In the case of list colorings, 
$f^{-1}(j)$ have to be a stable set of $G_j$, i.e. $f^{-1}(j)\in \ST_j$. 

However, 
to recover a list coloring from a partition of vertices into stable sets it is necessary to identify the color assigned to each stable set in the partition. 
For example, consider $G$ being a path of vertices 
$V=\{v_i: 1\leq i\leq 5 \}$ and edges $E=\{(v_i, v_{i+1}): 1\leq i\leq 4\} $, $\CC=\{1,2,3\}$, $V_1=\{v_1,v_5\}$,  $V_2=\{v_2,v_3,v_4\}$, and $V_3=\{v_1,v_2,v_3,v_4\}$. The partition of $V$ given by  the stable sets $S=\{v_1,v_5\}$, $S'=\{v_2,v_4\}$, and $S''=\{v_3\}$ defines two possible list colorings of $G$  since the vertices of $S'$ and $S''$ can interchange colors 2 and 3.

Thus, denoting by $X$ the set of pairs $(S,j)$ with $j\in \CC$ and $S\in \ST_j$, a list colorings can be characterized as a set $\hat{X}\subset X$ such that 
$\mathcal S=\{S: (S,j)\in \hat{X}, \text{ for some } j\in \CC \}$ is a partition of the set of vertices, for all $S \in \mathcal S$,
$|\{j: (S,j) \in \hat X\}|= 1$,  and,  for all $j\in \CC$, $|\{S: (S,j) \in \hat X\}| \leq 1$. 

This  characterization of list colorings 
becomes \emph{oversized} in some cases. 
The extreme case occurs for instances coming from the GCP. 
There, $V_j = V$ for all $j\in \CC$ and it is not necessary to decide which color is associated with each stable set of the partition. 
In other words, in instances of the GCP, all colors are \emph{indistinguishable} from each other.

Indistinguishable colors might also appear in general instances of the WLCP, and it is possible to take advantage of this characteristic.
Let us first formalize this concept. 

\begin{defi}
Given an instance $(G,\mathcal V,w)$ of the WLCP, we say that two colors $j, k\in \CC$ are \emph{indistinguishable}, and we denote it by $j\cong k$,
if $V_j = V_{k}$ and $w_j=w_{k}$.
\end{defi}

Observe that $\cong$ is an equivalence relation on $\CC$. From now on, $\{\CC_j : j \in K \}$ denotes the partition of $\CC$ into the equivalence classes defined by $\cong$, where $K$ is a subset of $\CC$ that has exactly one representative color for each equivalence class.
So, $K \cap \CC_j = \{j\}$ for
each $j \in K$.
We say that each representative color $j\in K$ has a \emph{multiplicity} $m(j) \doteq |\CC_j|$.
Also, for a given vertex $v$, we denote $k(v) \doteq |K \cap L(v)|$, which is the number of representative colors that can be assigned to $v$ in a list coloring.

Then, an instance of the WLCP can also be described by $(G, K, \mathcal{V}, m, w)$, where $\mathcal{V}=\{ V_j\subset V: j\in K\}$ and $m, w \in \ZZ_+^K$. 

Besides, the characterization of list colorings as partitions of vertices into stable sets can be reformulated as follows:

\begin{rem}\label{indist} 
Given an instance  $(G, K, \mathcal{V}, m, w)$ of WLCP, if $X$ is the set of pairs $(S,j)$ with $j\in K$ and $S\in \ST_j$, every list coloring $f$ can be represented by a set $\hat{X}\subset X$ such that 
$\mathcal S=\{S: (S,j)\in \hat{X}, \text{ for some } j\in K\}$ is a partition of the set of vertices, 
for all $S \in \mathcal S$,
$|\{j: (S,j) \in \hat X\}|= 1$,  and,  for all $j\in \CC$, $|\{S: (S,j) \in \hat X\}| \leq m(j)$. 
Moreover, $w(f)=\sum_{ (S,j)\in \hat X} w_j$.
\end{rem}

Observe that this last characterization of list colorings coincides with the classical one when it is applied to instances of the GCP. Indeed, since in these cases all colors are indistinguishable we have $|K| = 1$. W.l.o.g.~we can assume $K=\{1\}$, $m(1) = |\CC|$, and $\ST_1$ coincides with the set of stable sets of the graph.  

\medskip 

Let us introduce an example of the aforementioned concepts. Figure \ref{fig:C4a} shows an instance of the LCP where the lists of colors are shown in braces. It is easy to verify that 
we can consider $K = \{1,3,7\}$ and  
the equivalence classes given by $\mathcal{C}_1 = \{1,2\}$, $\mathcal{C}_3 = \{3,4,5,6\}$, $\mathcal{C}_7 = \{7\}$ with multiplicities $m(1)=2$, $m(3)=4$, and $m(7)=1$.
Their color graphs $G_1=G$, $G_3$ and $G_7$ are shown in Figures \ref{fig:C4b}, \ref{fig:C4c}, and \ref{fig:C4d}, respectively.

The list coloring $f$ given by $f(v_1) = 3$, $f(v_2)= f(v_3) = 4$ and $f(v_4) = 1$ and the list coloring $g$ given by $g(v_1) = 3$, $g(v_2)= g(v_3) = 5$ and $g(v_4) = 1$, having the same weight, are both represented by $\hat X = \{(\{v_1\},3),(\{v_2,v_3\},3),(\{v_4\}, 1)\}$. 

\begin{figure}
\centering

\subcaptionbox{Instance $(G,L)$.\label{fig:C4a}}{
\begin{tikzpicture}
\node[circle, fill, inner sep=2pt, label={[align=center]below: $v_3$\\ \scriptsize $\{1,2,3,4,5,6\}$}] (1) at (0,0) {};
\node[circle, fill, inner sep=2pt, label={[align=center]below: $v_4$\\ \scriptsize $\{1,2\}$}] (2) at (2,0) {};
\node[circle, fill, inner sep=2pt, label={[align=center]above: \scriptsize$\{1,2,3,$\\[-3pt] \scriptsize~$4,5,6\}$\\$v_2$}] (3) at (2,2) {};
\node[circle, fill, inner sep=2pt, label={[align=center]above: \scriptsize $\{1,2,3,4,$\\[-3pt] \scriptsize~$5,6,7\}$\\ $v_1$}] (4) at (0,2) {};
\draw (1) -- (2) -- (3) -- (4) -- (1);
\end{tikzpicture}
}
\subcaptionbox{$G_1 = G_2$.\label{fig:C4b}}{
\begin{tikzpicture}
\node[circle, fill, inner sep=2pt, label={[align=center]below: $v_3$\\ \scriptsize \color{white}$\{1\}$}] (1) at (0,0) {};
\node[circle, fill, inner sep=2pt, label={[align=center]below: $v_4$}] (2) at (2,0) {};
\node[circle, fill, inner sep=2pt, label={[align=center]above: $v_2$}] (3) at (2,2) {};
\node[circle, fill, inner sep=2pt, label={[align=center]above: $v_1$}] (4) at (0,2) {};
\draw (1) -- (2) -- (3) -- (4) -- (1);
\end{tikzpicture}
}
\subcaptionbox{$G_3 = G_4 = G_5 = G_6$.\label{fig:C4c}}[0.3\textwidth]{
\begin{tikzpicture}
\node[circle, fill, inner sep=2pt, label={[align=center]below: $v_3$\\ \scriptsize \color{white}$\{1\}$}] (1) at (0,0) {};
\node[circle, fill, inner sep=2pt, label={[align=center]above: $v_2$}] (3) at (2,2) {};
\node[circle, fill, inner sep=2pt, label={[align=center]above: $v_1$}] (4) at (0,2) {};
\draw (1) -- (4) -- (3);
\end{tikzpicture}
}
\subcaptionbox{$G_7$.\label{fig:C4d}}[0.1\textwidth]{
\begin{tikzpicture}
\node[circle, inner sep=2pt, label={[align=center]below: \color{white}$v_3$\\ \scriptsize \color{white}$\{1\}$}] (1) at (0,0) {};
\node[circle, fill, inner sep=2pt, label={[align=center]above: $v_1$}] (4) at (0,2) {};
\end{tikzpicture}
}

\caption{Example of a LCP instance.}
\label{fig:C4}
\end{figure}
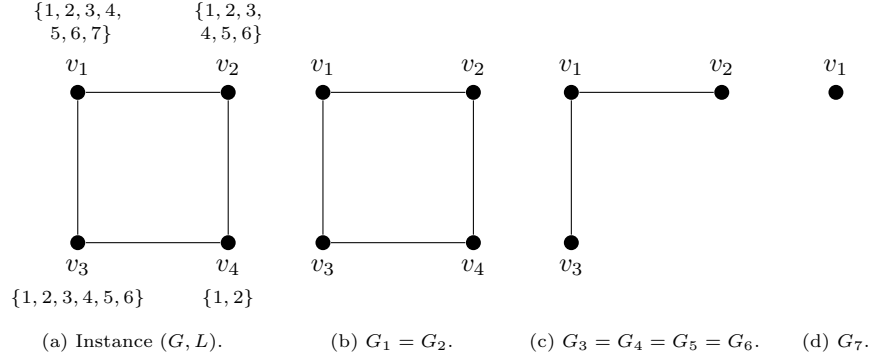


\medskip

In the following, we  present some remarks and results that will allow to reduce instances in a preprocessing stage.



As in the case of the Precoloring Extension Problem, there may be instances of the WLCP where some vertices have a single representative color in their lists.
Particular instances of WLCP where, for some $j \in K$, $k(v)=1$ for all $v\in V_j$ (equivalently, $V_j$ does not intersect any other representative color-vertex set), can be reduced according to the following remark.


\begin{rem} \label{porcolor}
Let $(G,L,w)$ be an instance of WLCP and $j\in K$ such that $k(v)=1$ for every $v\in V_j$. Clearly, if $m(j) < \chi(G_j)$, the instance is infeasible. Otherwise, solving WLCP on $(G,L,w)$ is equivalent to  solve 
 GCP on $G_j$ and WLCP on $(G-V_j, L|_{V'}, w|_{\mathcal{C}'})$ where $V'=V\setminus V_j$ and $\mathcal C'= \mathcal{C}\setminus \mathcal C_j$. The optimal list coloring of $(G,L,w)$ can be derived as the union of both optimal colorings. Therefore, if $f^*$ is an optimal list coloring of $(G-V_j, L|_{V'}, w|_{\mathcal C'})$, the optimal objective value of $(G,L,w)$ is $w_j.\chi(G_j)+w(f^*)$.  \end{rem}

It is also possible to reduce instances where only some  vertices have a single representative color in their lists. 
Let $j \in K$ and $v \in V_j$ with $k(v) = 1$.
Observe that every list coloring assigns to $v$ a color from $\CC_j$ but it is not relevant which one, since all of them are indistinguishable.
Therefore, we can take advantage of this fact and \emph{precolor} $v$ with $j$, since colors in $\CC_k \setminus \{j\}$ would lead to colorings with the same weight.
An instance where $v$ is precolored with $j$ can be thought of as a new instance of the WLCP without $v$ and without the possibility of coloring the neighbors of $v$ with $j$. 
Moreover, any isolated vertex $u$ of $G_j$ can also be precolored with $j$, since for all $w \in N(u)$, $w \notin V_j$.
These ideas can be generalized as follows:

\begin{lemma}\label{preproceso}
Let $(G,L,w)$ be an instance of WLCP with $G=(V,E)$, $j\in K$, and $Q$ be a maximal clique of the subgraph of $G_j$ induced by $\{v \in V_j : k(v) = 1\}$.
Then, if $m(j)<|Q|$, $(G,L,w)$ is infeasible. Otherwise, solving the WLCP on $(G,L,w)$ can be reduced to solve an instance of WLCP on some proper induced subgraph of $G$. 
\end{lemma}
\begin{proof}
Let $\tilde I_j = \{v \in V_j \setminus Q: N_{G_j}(v) \subset  Q , \,  N_{G_j}(v) \neq  Q \}$ and $V' = V \setminus (Q \cup \tilde I_j)$.  Observe that every vertex in $\tilde I_j$ is an isolated vertex in $G_j - Q$. 

Since $|Q|$ different colors of $\mathcal C_j$ are needed to color $Q$, the instance is infeasible if $m(j)<|Q|$. 
Otherwise, let $g$ be a coloring of $G[Q]$ with colors in $\mathcal C_j$.

 We consider the instance $(G', L', w')$ of WLCP where $G'=G[V']$, $L'(v) = L(v) \setminus g(Q \cap N(v))$ for all $v \in V'$, and $w'_k=0$ for all $k \in g(Q)$,  $w'_k=w_k$ otherwise. 
 
Let $f'$ be an optimal list coloring of $(G', L', w')$.
Defining $f(v) = g(v)$ for all $v \in Q$, $f(v) \in g(Q \setminus N_{G_j}(v))$ for all $v \in \tilde I_j$, and $f(v) = f'(v)$ for all $v \in V'$, a list coloring of $(G,L,w)$ is obtained with $w(f) = w_j.|Q|+w'(f')$.

To prove that $f$ is optimal, let $\tilde f$ be a list coloring of $(G,L,w)$. W.l.o.g. we can assume that $\tilde f(v)=g(v)$ for all $v\in Q$. Then,  $\tilde f|_{V'}$ is a list coloring of $(G', L', w')$ and $$w(\tilde f) = w(\tilde f|_{V'}) + w_j.|Q| \geq w(f') + w_j.|Q| = w(f).$$
Observe that we have also proved that if the instance $(G', L', w')$ is infeasible, then $(G,L,w)$ is also infeasible.
\end{proof}

For instances of WLCP coming from GCP, the well-known preprocess of coloring a maximal clique can be seen as applying this lemma.

\medskip
Given a feasible instance of WLCP, the previous lemma allows us to reduce in polynomial time to an instance where each vertex has at least two representative colors in its list. 

However, it is important to have in mind that, in each application of the lemma, the partition of $\mathcal{C}$ into indistinguishable colors might change. 
In particular, if $m(j) \geq |Q| + 1$, all the colors in $\CC_j\setminus g(Q)$ remain indistinguishable, whereas each $k\in g(Q)$ will belong to other equivalence class. 
Specifically, the new instance has:

\begin{enumerate}
\item for $k \in g(Q)$, $V'_k = V_k \setminus (N_{G_k}[g^{-1}(k)] \cup \tilde I_j)$ and $w'_k = 0$, and
\item for $k \in \CC \setminus g(Q)$, $V'_k = V_k \setminus (Q \cup \tilde I_j)$ and $w'_k = w_k$.
\end{enumerate}

In the example of Figure \ref{fig:C4}, $v_4$ can only be colored with $1 \cong 2$, w.l.o.g suppose the former is picked. 
The application of the above lemma with $Q=\{v_4\}$ conducts to $V'_1 = \{v_1\}$ and $V'_2 = \{v_1, v_2, v_3\}$, which means that colors 1 and 2 are no more indistinguishable.
Indeed, color 1 is now indistinguishable to color 7 and color 2 is now indistinguishable from colors 3, 4, 5, and 6.

\section{Set Covering formulation for WLCP} \label{ilpsc}

\citet{trick1996} present an ILP formulation for the GCP, based on the characterization
of colorings as partitions of vertices into stable sets.
We adapt these ideas to WLCP by considering the Remark \ref{indist}.

Given $(G, K, \mathcal{V}, m, w)$ an instance of WLCP with $G=(V,E)$ and  $X \doteq \{(S,k) : k \in K,~ S \in \ST_k \}$, a list coloring can be represented by binary variables $x_S^k$ for each $(S,k) \in X$, such that $x_S^k = 1$ if all vertices of $S$ are colored with the same color in $\CC_k$ (not necessarily $k$ itself).
Then, the WLCP is formulated as follows:

\begin{align}
\textrm{(WLCP-SC)} \hspace{20 pt} \min \sum_{(S,k)\in X} w_k &  x^k_S  &  \notag \\
s.t. \hspace{47pt} & & \notag \\
\sum_{\mathclap{\substack{(S,k) \in X:\\v \in S}}} \hspace{5pt} x^k_S & \geq 1 &   
  ~v \in V \;  \label{SCRESTR1X}\\[7pt]
\sum_{S \in \ST_k} x^k_S & \leq m(k) & 
  ~k \in K \; 
  \label{SCRESTR2X} \\[7pt]
x^k_S & \in \ZZ_+ & 
	~(S,k)\in X  \label{SCRESTR3X}
\end{align}

Constraints \eqref{SCRESTR1X} assert that, for each vertex $v$, at least one stable set of a color graph containing $v$ is selected. We refer to them as \emph{vertex constraints}.  
Constraints \eqref{SCRESTR2X} allow taking, for each representative color $k\in K$, at most $m(k)$ stable sets of $G_k$.
We refer to them as \emph{color constraints}. 
Regarding the integrality constraints of the variables, since we consider non-negative weights, there is always an optimal solution such that $x^k_S \in \{0,1\}$ for all $(S,k)\in X$.
For feasible instances, integrality constraints in \eqref{SCRESTR3X} and non-negative weights guarantee an optimal solution such that $x^k_S \in \{0,1\}$ for all $(S,k)\in X$.

Observe that, if $m(k) \geq |V_k|$ for some $k\in K$, the corresponding color constraint becomes redundant and can be omitted. 
In particular, for instances coming from the GCP there is a single representative color and its corresponding color constraint is redundant.  
Then, for such instances, WLCP-SC coincides with the set covering formulation for the GCP given in \cite{trick1996}.

Like the set covering formulation for the GCP, the number of variables of WLCP-SC may grow exponentially with the number of vertices of $G$ and, for medium or large instances, it may become
intractable to perform an exhaustive enumeration of them. 
Branch-and-price schemes are the usual way to address the resolution of such models.

The two main components of a $B\&P$ algorithm are the column generation routine to solve lineal relaxations and the branching strategies.
The latter should be \emph{robust}, i.e. the subproblems should also be instances of the WLCP, to reuse the column generation routine.
For a background of this technique in coloring problems, we refer the reader to \cite{trick1996,held2011,toth2011,malaguti2012,hertz2014}. 


In the following sections, we present our proposals for these two components.

\section{Solving the linear relaxation}\label{Sec:Solving:LR}

Given $(G, K, \mathcal{V}, m, w)$ an instance of WLCP with $G=(V,E)$ and  $X \doteq \{(S,k) : k \in K,~ S \in \ST_k \}$, let $LR$ be the linear relaxation of WLCP-SC. 
By multiplying both sides of the color constraints by $-1$, for every $(S,j)\in X$, the column corresponding to variable $x^j_S$ in the coefficient matrix can be written as $(x^S,-e_j)$ where $x^S\in \{0,1\}^V$ is the characteristic vector of $S$ and $e_j\in \{0,1\}^K$ is the $j^{\text{th}}$ canonical vector.

As mentioned, $LR$ can have an exponential number of variables, which needs to be addressed with column-generation techniques.

\subsection{Pricing Routine}

For any $\hat X \subset X$, 
let $LR(\hat X)$ be the linear program obtained from $LR$, by setting all variables in $X\setminus \hat X$ to zero.
Feasible solutions of $LR(\hat X)$ will be denoted as $(\hat x, \mathbf 0)$ with $\hat x\in \RR^{\hat X}$

Let $(\hat x^*, \mathbf 0)$ be an optimal solution of $LR(\hat X)$ and $(\pi^*, \gamma^*)$ be the optimal solution of the dual of $LR(\hat X)$,
where $\pi\in \RR^V$ and $\gamma \in \RR^{K}$ are the dual variables corresponding to vertex and color constraints, respectively.

Then, for any $(S,j)\in X$, the reduced cost of variable $x^j_S$ in $LR$ is
$$c^j_S \doteq  w_j- (\sum_{v\in S} \pi^*_v - \gamma^*_j).$$

In order to know if $(\hat x^*, \mathbf 0)$ is an optimal solution of $LR$, we have to decide if there exists $(S,j)\in X$ such that
$c^j_S < 0$. Then, the pricing routine has to answer the following question:
\begin{center}
\emph{Is there}~$(S,j) \in X \setminus \hat X$ \emph{such that} $\sum_{v\in S} \pi^*_v > \gamma^*_j + w_j $?
\end{center}

The pricing routing reduces to solve, for each $j\in K$, the \emph{Maximum Weighted Stable Set Problem} (MWSSP) 
on $G_j$, with weight
$\pi^*_v$ for each $v\in V_j$. For the case $|K|=1$, it coincides with the one proposed in \cite{trick1996} for GCP. 

Despite the MWSSP being $\mathcal{NP}$-hard, the enumeration routine given in \cite{held2011} shows to be reasonably fast in the context of a B\&P algorithm (see e.g.~\cite{held2011,malaguti2012}). 

\medskip

The enumeration routine returns the first stable set whose weight is greater than a given threshold, and when this is not possible, it returns the stable set with the largest weight.
Our implementation uses a threshold $\beta(\gamma^*_j + w_j)$ with $\beta = 1.1$ (the same value is considered e.g. in \cite{trick1996,malaguti2012}).
At the same time, we avoid the execution of the enumeration routine when there are not enough vertices in $G_j$ to find a stable $S$ satisfying $c^j_S < 0$, e.g.~if $\sum_{v \in V_j} x^*_v \leq w_j + \gamma_j^*$.

Based on preliminary experiments, incorporating many columns at once before re-optimizing decreases the number of iterations of the process and, thus, the overall time.
In our implementation, we search for at most one entering variable $(S,j)$ for each $j \in K$.
Besides, stable sets are greedily expanded to maximal ones in their corresponding color graphs before adding the new variable to $\hat X$.

\subsection{Initialization}\label{inicializacion}

The column generation process needs to start with a set $\hat X\subset X$ for which $LR(\hat X)$ is feasible. 
Unlike GPC,  
such a  set $\hat{X}$ might not exist, i.e. $LR$ may not be feasible.




\medskip

This particularity forces us to design an initialization procedure which either outputs a set $\hat X\subset X$ such that $LR(\hat X)$ is feasible or says that $LR$ is not feasible. 

The initialization procedure works with an \emph{extended formulation} 
which corresponds to an \emph{extended instance} obtained by adding, for each $v\in V$, a \emph{dummy} color, referred to as $k_v$, with color-vertex set $V_{k_v} = \{v\}$ (i.e. $v$ is the unique vertex having $k_v$ in its color list). We denote  the linear relaxation of this extended formulation by $LR_{ext}$.
As $G_{k_v}$ has only one non-empty stable set, $LR_{ext}$ has exactly $|V|$ more variables than $LR$, precisely those
$x^{{k_v}}_{\{v\}}$ for each $v \in V$, referred to as \emph{dummy variables}.

By setting the value of dummy variables to 1 and the others to 0, a feasible solution is always obtained for $LR_{ext}$. 
Hence, $\hat X = \{(\{v\}, k_v) : v \in V\}$ is a valid initial set to start the column generation process to solve $LR_{ext}$.

The relationship between the optimal solution of $LR_{ext}$ and the feasibility of $LR$ depends on the definition of an adequate objective function for $LR_{ext}$, as explained below.

The traditional approach considers the weight of dummy colors equal to 1
and zero for the remaining colors. In this case, $LR$ is feasible if and only if the optimal value of
$LR_{ext}$ is zero.
Moreover,  the column generation procedure for $LR$ can start with $\hat X$ being the set of basic variables of the basic optimal solution of $LR_{ext}$ (where all dummy variables are not basic). 
However, the optimal objective value of $LR_{ext}$ could be very far from $LR$'s, as the optimization of $LR_{ext}$ does not take advantage of the original weights.

Other known approach is to consider
a \emph{big-$M$} value for the weight of dummy colors, i.e. $w_{k_v} = M$ for all $v \in V$, and the original weight $w_k$ for the other colors $k \in K$. We call $LR_{ext}^M$ the linear relaxation of this auxiliary problem. 
The underlying question in this approach is if there exists a value of $M$ such that $LR$ is feasible if and only if $LR_{ext}$ has an optimal solution with all dummy variables at zero. 
Indeed, if such $M$ exists, this optimal solution of $LR_{ext}$ is (by deleting dummy variables) an optimal solution of $LR$.
The following result allows us to prove the existence of such $M$. 


\medskip

\noindent {\bf{Hadamard's inequality  (\cite{Garling2007}):}}\\
Let $B$ be a $n\times n$
$\{-1,0,1\}$-matrix, then $|det(B)|\leq n^{n/2}$.

\medskip

Then, we have the following result:
\begin{lemma}
There exists $M>0$ such that $LR$ is feasible if and only $LR_{ext}^{M}$ 
has an optimal solution with all dummy variables at zero.
Moreover, if $(x^*,\bf{0})$ is a basic optimal solution of $LR_{ext}^{M}$ (where $\bf{0}$ is the vector whose components are the values of dummy variables), then $x^*$ is a basic optimal solution of $LR$.
\end{lemma}

\begin{proof}

Assume that bases of $LR$ have size $n\times n$ and let $H = n^{n/2}$.  By Hadamard's inequality we have that, for all $M>0$, if $B$ is a base of $LR_{ext}^{M}$ then $|det(B)|\leq H$. Besides, if $LR$ is feasible, its  objective function is bounded by $W := \sum_{k\in K} w_k m_k$.

Let $M = H (W+1)$, $B$ be an optimal base of  $LR_{ext}^{M}$ and 
$(x^*,d^*)$  the corresponding  optimal basic solution (where $d^*$ is the vector whose components are the values of dummy variables), and  $z^*_{ext}= \sum_{(S,j)\in X} w_j x^{*S}_j + M \sum_{v \in V} d^*_v$,  the optimal value of $LR_{ext}^M$.

Clearly, if $d^* = \mathbf 0$, then $x^*$ is a feasible solution of $LR$ and therefore, $LR$ is feasible.
To see the converse, assume that, for some $v \in V$, $d^*_v > 0$. 
By Cramer's rule, we have  $$d^*_v=|d^*_v| \geq \frac{1}{|det(B)|} \geq \frac{1}{H}$$ 
and then, $$z^*_{ext}\geq M d^*_v  \geq \frac{M}{H}=W+1.$$
Assume that $LR$ is feasible, i.e. there  exists a feasible solution $\tilde x$ of $LR$ with objective value $\tilde z$. Then,  $(\tilde x, \mathbf 0)$ is a feasible solution of  $LR^M_{ext}$ with objective value
$\tilde z$. Then, $\tilde z \geq z^*_{ext} \geq W+1$. 
Recalling that $W$ is an upper bound of the objective function of $LR$, it leads to a contradiction. 
Then, $LR$ is not feasible.


\end{proof}

Setting $M$ to the value proposed in the previous lemma, a basic optimal solution of $LR_{ext}$ tell us if $LR$ is feasible and, in that case, an optimal solution of $LR$ can directly be obtained. 

However, this value of $M$ could be too big, giving rise to numerical instability issues. 
Despite this, for an arbitrary $M > 0$ and considering $W$ defined as in the proof of the previous lemma, we have the following:

\begin{rem}\label{remark.M}
Let $(x^*,d^*)$ and $z^*_{ext}$ be an optimal solution and the optimal value of $LR_{ext}^M$, respectively. 
Then:
\begin{enumerate}
    \item If $d^* = \mathbf 0$, then $x^*$ is an optimal solution of $LR$.
    \item If $z^*_{ext} > W$, then $LR$ is infeasible.
\end{enumerate}
\end{rem}

Our implementation consider $M = 1000$ for the initialization procedure.
Such a value works reasonably well for the instances considered in the computational experiences.
Indeed, in the hundreds to thousands of linear relaxations solved through the B\&P, the remaining case of Remark \ref{remark.M} never happened, i.e. when the optimal solution of $LR_{ext}$ has some non-zero dummy variable and $z^*_{ext} \leq W$. 

\section{Dividing the set of solutions}\label{Sec:branchingstrategies}

In the literature on exact resolution of coloring problems, there are two main strategies to divide the set of solutions: 
branching on edges \cite{trick1996,toth2011, malaguti2012, hertz2014}
and branching on colors \cite{brelaz1979, mendezdiaz2006}.
In this section, we present how to deal with these strategies in the context of the WLCP.

\medskip

\noindent \emph{Branching on edges.}
This strategy was first proposed for GCP in \cite{zykov1949} and then used in \cite{trick1996}.
Given a graph $G$, the idea is to pick two non-adjacent vertices $u$ and $v$ and divide the solution space into those colorings where $u$ and $v$ have the same color and those where they do not.
Finding the optimal coloring in the latter case is equivalent to solve a GCP on the graph obtained by adding the edge $(u,v)$ to $G$, denoted by $G + (u,v)$.
The former case is equivalent to solve a GCP on the graph obtained from $G$ by collapsing vertices $u$ and $v$ into a single one, named $uv$.
That is, vertices $u$ and $v$ are removed and a new vertex $uv$ is added and connected to every vertex from
$N_G(u) \cup N_G(v)$.
The resulting graph will be referred to as $G_{uv}$.

\medskip{}{}{}{}

When adapting this strategy for the WLCP, the vertices $u$ and $v$ must verify the additional condition of sharing at least one color in their lists of colors, i.e.~$L(u) \cap L(v) \neq \emptyset$.
Otherwise, the subproblem on $G_{uv}$ would be infeasible, as it is not possible to assign the same color to both vertices, and, by Lemma \ref{uselessedges}, the subproblem on $G + (u,v)$ would be equivalent to the original one.

To see that this branching is robust for the WLCP, let $(G,L,w)$ be an instance and let $u$ and $v$ be two non-adjacent vertices of $G$ such that $L(u) \cap L(v) \neq \emptyset$, or equivalently, $u$ and $v$ are two non-adjacent vertices in $G_j$ for some $j \in \CC$.
It is not hard to see that, on the one hand, finding an optimal list coloring $f$ with $f(u) \neq f(v)$ is equivalent to solve the WLCP on the instance $(G+(u,v), L, w)$. 
On the other hand, finding an optimal list coloring $f$ with $f(u) = f(v)$ is equivalent to solve the WLCP on the instance $(G_{uv},L',w)$ where $L'(z) = L(z)$ for all $z \in V(G_{uv}) \setminus \{uv\}$ and $L'(uv) = L(u) \cap L(v)$.

The convergence of this branching strategy is guaranteed since an instance where, for all $j \in \mathcal{C}$, $G_j$ is a complete (possibly trivial) graph is reached in a finite number of applications.
There, the WLCP can be solved in polynomial-time by Proposition \ref{completeispoly}. In our implementation, instead of using the Hungarian algorithm, we solve the linear relaxation through the column generation process, which outputs an integer solution (since the coefficient matrix of LR$_{ext}$ is totally unimodular). In addition, the enumeration tree is explored following a depth-first search, starting with the subproblem where the vertices collapse.

\medskip

\noindent \emph{Branching on colors.}
This strategy for GCP appears in direct enumeration algorithms, such as DSATUR \cite{brelaz1979}, and later in \cite{mendezdiaz2006}.
Given a graph $G$, the idea is to choose a vertex $v$ and a color $j$, and then divide the solution space between those colorings where $v$ is colored with $j$ and those where it is not.
Additionally, this branching can easily be made non-dichotomous by  coloring $v$ with a different color in each subproblem. 

 It is not straightforward to incorporate this branching strategy into a B\&P algorithm since it is not robust for GCP, i.e. the subproblems to solve after branching do not correspond to GCP instances.
However, they can be seen as LCP instances. 
Indeed, finding the optimal coloring such that $v$ is colored with $j$ is a Precoloring Extension Problem (which is a particular case of LCP, as mentioned in Section \ref{intro}).
Similarly, finding the optimal coloring such that $v$ is not colored with $j$ can also be modeled as a LCP where $j$ is removed from the list of colors of $v$.

\medskip

When adapting this strategy for the WLCP, it is possible to take advantage of indistinguishable colors. The solution space can be  partitioned into those list colorings where $v$ is colored with some color of $\CC_j$ and those where it is not.
This only makes sense when $k(v)\geq 2$, otherwise the subproblem where $v$ is colored with some color of $\CC_j$ would be equivalent to the original problem and the other one would be infeasible.
For this reason, in every node of the branching tree, instances are always preprocessed according to Lemma \ref{preproceso} until each vertex has more than two representative colors in its list.


To see that this branching is robust for the WLCP, let $(G,L,w)$ be an instance, $j\in K$, and $v\in V_j$ with $k(v) \geq 2$. 
It is easy to prove that, on the one hand, finding an optimal list coloring such that $v$ is colored with some color of $\CC_j$ is equivalent to solve a WLCP on $(G,L_1,w)$ where $L_1(v) = \CC_j$ and $L_1(u) = L(u)$ for all $u \in V(G) \setminus \{v\}$.
On the other hand, finding an optimal list coloring such that $v$ is not colored with any color of $\CC_j$ is equivalent to solve a WLCP on $(G,L_2,w)$ where $L_2(v) = L(v) \setminus \CC_j$ and $L_2(u) = L(u)$ for all $u \in V(G) \setminus \{v\}$.

Observe that $k(v)$ is reduced by 1 in the instance $(G,L_2,w)$ and $k(v) = 1$ in the instance $(G,L_1,w)$.
Thus, the strategy eventually reaches a vertex $v$ with $k(v) = 1$, which will be prepossessed, resulting in a graph with fewer vertices.
These facts prove the convergence of this branching strategy since a trivial graph is reached in a finite number of applications.

\medskip

Again, our implementation explores the enumeration tree following a depth-first search, this time starting with the subproblem where a representative color is assigned to a vertex.

\section{Variable selection for branching} \label{Sec:branchingvariable}

In this section, different strategies are proposed to select the variables involved in each of the branching strategies presented in the previous section. For the remainder of the section, $x^*$ denotes an optimal fractional solution of LR.
In our implementation, unless another criterion is specified, we choose the most fractional variable.

\paragraph{Branching on edges}

Recall that this branching strategy requires selecting two vertices $u$ and $v$ that are not adjacent in some color graph, say $G_k$. We propose to select two vertices belonging to a stable set $S\in \mathcal S_k$ such that $x^{*k}_S$ is fractional. The following lemma proves the existence of candidates variable for this branching.

\begin{lemma} \label{lemita}
Let $x^*$ be a basic fractional optimal solution of $LR$.
Then, there exists $(S,k)\in X$ with $|S|\geq 2$ such that $x^{*k}_S$ is fractional. 
\end{lemma}
\begin{proof}
Let $X' = \{(S,k) \in X : |S|\geq 2 \}$ and suppose that $x^{*k}_S\in \{0,1\}$ for all $(S,k)\in X'$.
Let $LR^\dagger$ be the linear program obtained by replacing $x^{k}_S$ with the value $x^{*k}_S$, for all $(S,k)\in X'$.  Clearly, 
$x^*$ is also an optimal solution of $LR^\dagger$.
Moreover, all variables $x_S^k$ of $LR^\dagger$ satisfy $|S|=1$.
Then, the columns of its coefficient matrix has only two non-zero entries: one entry equal to 1 in the row corresponding to the unique vertex in
$S$ and one entry equal to $-1$ in the row corresponding to color $k$.
Therefore, the coefficient matrix of $LR^\dagger$ is totally unimodular and, as the r.h.s.~of constraints are integer, the basic optimal
solutions of $LR^\dagger$ are integer too, in particular $x^*$, which is a contradiction.
\end{proof}

Clearly, by selecting arbitrary vertices $u$ and $v$ from some stable set $S$ according to the previous lemma, $S$ can no longer be part of a feasible list coloring of $G + (u,v)$.
But this might not be true for the subproblem where $u$ and $v$ receive the same color.
For example, if $|\hat S \cap \{u,v\}| \in \{0,2\}$ for all $(\hat S, \ell) \in X$ such that $x^{*\ell}_{\hat S} > 0$, then an equivalent feasible solution can be obtained for $G_{uv}$ by replacing $u$ and $v$ by $uv$ in those stable sets that contains both $u$ and $v$.
Thus, two more sophisticated strategies are proposed.

\begin{itemize}
    
\item[$\blacktriangleright$] \textsc{Edg-Std}. This strategy extends the standard one proposed by \cite{trick1996} for GCP (in fact, it behaves exactly the same for instances coming from GCP):

\begin{enumerate}[i.]
\item Find $(S,k) \in X$ such that $|S| \geq 2$ and $0 < x_{S}^{*k} < 1$. \label{EDGE-STD-1}
\item Choose $u \in S$. \label{EDGE-STD-2}
\item Find $(\hat S, \ell)\in X\setminus \{(S,k)\}$ such that $u\in \hat S$ and $x_{\hat S}^{*\ell} > 0$. \label{EDGE-STD-3}
\item If $\hat S \neq S$, choose $v \in S \triangle \hat S$. Otherwise, choose $v \in S\setminus\{u\}$. \label{EDGE-STD-4}
\end{enumerate}

\end{itemize}

The existence of such $(\hat S, \ell)$ is guaranteed since $0 < x^{*k}_{S} < 1$ and the vertex constraint must hold for $u$.

When $k = \ell$ (as it happens on instances coming from GCP), the vertex $v$ can always be selected from $S \triangle \hat S$ since the stable sets in $\ST_k$ are all different.
Even more, if the stable sets are maximal, then $v$ must be adjacent to some vertex of the other stable set.
W.l.o.g. assume that $v \in \hat S \setminus S$ and $v$ is adjacent to $w \in S$. Then, in $G_{uv}$, $uv$ and $w$ are adjacent and $S$ can no longer be part of a feasible list coloring of $G_{uv}$.

However, for general instances of WLCP, $k \neq \ell$ might occur.
In this case, if $v$ is picked from $S \triangle \hat S$ and the stable sets are maximal, then it is also possible to prove that $x^*$ is cut (by observing that $v$ must be adjacent to some vertex of the other stable set if $v$ belongs to both $V_k$ and $V_\ell$, whereas $V_k$ or $V_\ell$ does not contain the vertex $uv$ otherwise).

Unfortunately, $x^*$ might not be cut when selecting $v$ from $S \setminus \{u\}$.
Although a different stable set $\hat S$ can be considered in that case, being $S = \hat S$ is very infrequent and obtaining the same fractional solution $x^*$ after branching can only happen a finite number of times (as $G_{uv}$ has one less vertex).

Observe that the branching on edges has a greater structural impact when $u$ and $v$ belongs to many vertex-color sets.
Recall that, in $G + (u,v)$, the vertices $u$ and $v$ become adjacent in every color graph that contains both vertices; whereas in $G_{uv}$, the vertices $u$ and $v$ are removed from every color graph that contains only one of them.

\medskip

An alternative strategy is proposed prioritizing vertices that cause the greatest impact when branching.

\begin{itemize}
\item[$\blacktriangleright$] \textsc{Edg-alt}. Find the vertices $u$ and $v$ that maximize $k(u) + k(v)$ subject to there exists $(S,k) \in X$ with $\{u,v\} \subset S$ and $0 < x^{*k}_{S} < 1$.
\end{itemize}

\paragraph{Branching on colors}
In this branching strategy, a vertex $v$ and a representative color $k$ such that $v\in V_k$ must be selected.
Recall that according to Proposition \ref{preproceso}, $k(v) \geq 2$  for every vertex $v$.

Three strategies are proposed to select $v$ and $k$. The first strategy follows part of the idea of \cite{trick1996}:

\begin{itemize}
\item[$\blacktriangleright$] \textsc{Clr-Std}. 
\begin{enumerate}[i.]
    \item Find $(S,k) \in X$ such that $0 < x^{*k}_{S} < 1$.
    \item Choose $v \in S$.
\end{enumerate}
\end{itemize}

 Observe that the branching on colors has a greater structural impact when $v$ has multiple neighbours in $G_k$. 
Recall that the Preposition \ref{preproceso} can be applied to the subproblem where $v$ must receive some color of $\mathcal{C}_k$, and thus the vertices in $N_{G_k}(v)$ are removed from $G_k$.
Besides, the vertices with a lower value of $k(v)$ are usually the most critical to color, as the few candidate colors can quickly become unavailable if their neighbors are colored first.
At the same time, they give much more information to the subproblem where $v$ cannot receive a color of $\mathcal{C}_k$.
Under this consideration, two alternative strategies are presented: 

\begin{itemize}
\item[$\blacktriangleright$] \textsc{Clr-Alt$_1$}. Find the vertex $v$ and $k\in K$ that maximize $|N_{G_k}(v)|$ subject to there exists $(S,k) \in X$ with $v \in S$ and $0 < x^{*k}_{S} < 1$. 
In case of a tie, choose the vertex $v$ with the lowest $k(v)$ and then, the color $k$ with the lowest $m(k)$.
\end{itemize}

\begin{itemize}
\item[$\blacktriangleright$] \textsc{Clr-Alt$_2$}:
\begin{enumerate}[i.]
    \item Find a vertex $v$ that minimizes $k(v)$ and such that exists $(S,k) \in X$ with $v \in S$ and $0 < x^{*k}_{S} < 1$.
    \item Find the representative color $k$ that maximizes $|N_{G_k}(v)|$ and such that $0 < x^{*k}_{S} < 1$ for some $S \in \ST_k$ with $v \in S$.
    In case of a tie, choose the color $k$ with the lowest $m(k)$.
\end{enumerate}
\end{itemize}




Observe that, whenever $v$ belongs to some stable set $S$ with fractional $x^{*k}_{S}$, the fractional solution $x^*$ is cut from the feasible region of the subproblem where $v$ cannot receive a color of $\mathcal{C}_k$.
But this might not be enough for cutting $x^*$ in the feasible region of the other subproblem.

\section{Sets of instances}\label{Sec:Instances}

In this section, the set of instances used for the computational experiments are presented.

There are two types of instances widely used for testing graph coloring algorithms:
random and benchmark.
In the case of GCP, two parameters are involved in the generation of random instances: $n$ (number of vertices) and $p$ (edge probability, i.e. probability for two distinct vertices to be adjacent), and the best-known benchmark instances come from the DIMACS Challenge \cite{DIMACS}.
In our case, two difficulties arise: parameters $n$ and $p$ are insufficient, and no benchmark instances are available for WLCP.
We propose the following:

\subsection*{Random instances for WLCP}

The parameters involved in the generation of random instances for WCLP are the following:

\begin{itemize}
    \item $n \in \NN$: number of vertices,
    \item $p \in (0, 1)$: edge probability,
    \item $|K| \in \NN$: number of distinguishable colors,
    \item $m \in \NN^K$: vector of multiplicities,
    \item $w \in \NN^K$: vector of weights,
    \item $q \in (0, 1)$: vertex-color set probability, i.e. probability for a vertex of $G$ to belong to $V_k$ for each $k \in K$.
\end{itemize}


The generation procedure is as follows. First, a random graph is generated with parameters $n$ and $p$.
For each $k \in \{1, \ldots, |K|\}$, a color $k$ is considered with multiplicity $m(k)=m_k$ and weight $w_k$. Each initially empty set $V_k$ is filled with  vertex $v \in V$ with probability $q$.

Observe that $n$ and $q$ indirectly determine the size of the vertex-color sets and the overlap between them.
In fact, $|V_k| \simeq qn$ for each $k \in K$ and $k(v) \simeq q|K|$ for each $v \in V$. 
Besides, the color graphs are expected to have in average the same density than $G$ due to the randomness of the generation procedure. 


\begin{itemize}
    \item[$\blacktriangleright$] \emph{\textbf{Set 1}} (135 instances) consists of 5 instances for each combination of $p, q \in \{0.25, 0.5, 0.75\}$, generated from random seeds. The other parameters values are $n = |K| = 60$, $w = m = \mathbf{1}$.  
    Besides, 2 copies of each instance are made by changing $m = k \mathbf{1}$ with $k \in \{3,5\}$ (to preserve the structure of the instance).
\end{itemize}

The vector of weights is always $\mathbf{1}$ since weights of colors distort the structure of the instance, hiding the incidence of the rest of the parameters.
The study of different weights is left for future work.

\subsection*{Random instances for $\mu$-coloring problem}

A set of random instances for the $\mu$-coloring problem is considered to analyse the performance of the proposed algorithms over instances with different structures.
The following parameters are involved in the generation procedure:
\begin{itemize}
    \item $n \in \NN$: number of vertices,
    \item $p \in (0, 1)$: edge probability,
    \item $|K| \in \NN$: number of distinguishable colors,
    \item $m \in \NN^K$: vector of multiplicities,
    \item $t \in (0, 1)$: probability of removing a vertex between one color graph and the next one.
\end{itemize}

The generation procedure is similar to the previous one, unless each vertex-color set is generated based on the previous one.
This means that $V_1 = V$ and, for each $k \in \{2, \ldots, |K|\}$, $V_{k}$ has all the vertices from $V_{k-1}$ with probability $1-t$ (colors with empty color-vertex sets are discarded).

\begin{itemize}
    \item[$\blacktriangleright$] \emph{\textbf{Set 2}} (54 instances) consists of 3 instances for each combination of
    $n \in \{ 70, 80 \}$, $p \in \{ 0.25, 0.5, 0.75 \}$ and $t \in \{ 0.05, 0.1, 0.2 \}$. The other parameter values are  $|K| = \nicefrac{n}{2}$, $m(1) = \omega(G)$ (size of a maximum clique in $G$) and, for each $k \in \{2,\ldots,|K|\}$, $m(k)$ is a random integer in $[1,\ldots,5]$.  
    The choice of the multiplicity for the first color intends to mitigate the possibility of generating infeasible instances.

    As instances of WLCP, they have $w = \mathbf{1}$, $V_{|K|} \subseteq \ldots \subseteq V_1$, and thus $k(v)$ is variable: for each $k \in \{1,\ldots,|K|-1\}$, a vertex $v \in V_k \setminus V_{k+1}$ has $k(v) = k$, and a vertex $v \in V_{|K|}$ has $k(v) = |K|$.
    Indeed, the average $|V_k|$ and the average $k(v)$ decreases when $t$ grows.
    
\end{itemize}

\subsection*{Random instances for GCP}

A set of random instances for GCP is also considered to analyze the performance of the proposed algorithms.
\begin{itemize}
    \item[$\blacktriangleright$] \emph{\textbf{Set 3}} (50 instances) consists of 5 instances for each combination of:
    $n \in \{ 70, 80 \}$, $p \in \{ 0.1, 0.3, 0.5, 0.7, 0.9 \}$.
    
    As instances of WLCP, they have 
    $w = \bar{\mathbf{1}}$, $K = \{1\}$, $m(1) = \Delta(G)+1$, and $|V_1| = V$.
\end{itemize}

\subsection*{Benchmark instances for GCP}

Besides the random instances for GCP, a set of instances from the literature is considered.
\begin{itemize}
    \item[$\blacktriangleright$] \emph{\textbf{Set 4}} (34 instances) contains some of the instances from the DIMACS Challenge for the GCP \cite{DIMACS}.
    
    Specifically, we consider the most challenging instances from \citet{toth2011}, i.e. those where the lower bounds given by the clique number of the graph and the linear relaxation of the IP model differ from the upper bound given by the heuristic algorithm MMT.
    Besides, we excluded instances whose linear relaxation cannot be solved within the time limit.
\end{itemize}

\subsection*{Particular instances of WLCP}

Finally, particular instances of WLCP with different structures are considered.

\begin{itemize}
    \item[$\blacktriangleright$] \emph{\textbf{Set 5}} (16 instances) contains 3 instances of Precoloring Extension from \citet{quasigroup}, 10 instances of MWSCP from \citet{SETCOVER}, and 3 instances of WLCP from \citet{ejs2018}. These last instances are the only ones that consider non-unit weights and they correspond to real-world costs. Table \ref{Tab:set5} shows the list of instances considered in this set.
    The columns report the name, the problem, the number of vertices $n$ and edges $m$, the density $p$, the number of representative colors $|K|$, and the average $k(v)$ of the vertices of each instance.
\end{itemize}

\begin{table}[h]
    \centering
    \footnotesize
    \setlength{\tabcolsep}{6pt}
    \begin{tabular}{lrrrrrr}
    \toprule
        Name & Problem & $n$ & $m$ & $p$ & $|K|$ & $\bar k$ \\ \midrule
        \texttt{order18holes120} & Precol. Ext. & 324 & 5508 & 0.11 & 18 & 7 \\
        \texttt{order30holes316} & Precol. Ext. & 900 & 26100 & 0.06 & 30 &  11 \\
        \texttt{order30holes320} & Precol. Ext. & 900 & 26100 & 0.06 & 30 & 11 \\
        \texttt{scp41} & MWSCP & 200 & 0 & 0 & 1000 & 20 \\
        \texttt{scpcyc06} & MWSCP & 240 & 0 & 0 & 192 & 4 \\
        \texttt{scpa1} & MWSCP & 300 & 0 & 0 & 3000 & 60 \\
        \texttt{scpb1} & MWSCP & 300 & 0 & 0 & 3000 & 150 \\
        \texttt{scpc1} & MWSCP & 400 & 0 & 0 & 4000 & 80 \\
        \texttt{scpclr1} & MWSCP & 511 & 0 & 0 & 210 & 26 \\
        \texttt{scpcyc07} & MWSCP & 672 & 0 & 0 & 448 & 4 \\
        \texttt{rail507} & MWSCP & 507 & 0 & 0 & 62182 & 797 \\ 
        \texttt{rail516} & MWSCP & 516 & 0 & 0 & 47311 & 610 \\ 
        \texttt{rail582} & MWSCP & 582 & 0 & 0 & 54341 & 676 \\
        \texttt{trips215drivers60} & WLCP & 215 & 9442 & 0.41 & 3389 & 142 \\
        \texttt{trips587drivers136} & WLCP & 587 & 65106 & 0.38 & 19959 & 434 \\
        \texttt{trips818drivers191} & WLCP & 818 & 119130 & 0.36 & 36761 & 702 \\ \bottomrule
    \end{tabular}
    \caption{Instances of Set 5.}
    \label{Tab:set5}
\end{table}

\section{Computational experiences}\label{Sec:computationalexperiences}

From now on, our B\&P algorithm is called BP-LCol.
The objective of this section is to evaluate the performance of BP-LCol on Sets 1 to 5 and to compare the proposed branching rules and variable selection strategies.

The setup includes a computer equipped with an Intel i7 3.6GHz (a single thread is used), the operating system Ubuntu 18.04, GCC 7.5.0, and the solver IBM ILOG CPLEX 12.7 \cite{CPLEX}.
For further comparisons, benchmark program \emph{dfmax} reports an average of
2.83 seconds (user time) over 10 runs on the instance r500.5.b (freely available at \cite{DIMACS}).

An instance is considered \emph{solved} if the solver finds an optimal solution and proves it to be optimal within 2 hours (7200 seconds) of CPU time.

\subsection{Random instances of WLCP}

The results obtained for the random instances of the WLCP (Set 1) are presented in Table \ref{Tab:set1}.
The first 3 columns report the values of parameters $p$, $q$, and $m$, respectively.
Each row summarizes the results obtained for the 5 instances with those parameters, except for the last row which considers all the instances.
Column ``$\bar k$'' reports the average $k(v)$ of the vertices of the graph.
Column ``solved (\%)'' is the percentage of solved instances for each configuration.
Column ``time (s)'' is the average execution time (in seconds) of the solved instances for each configuration, and a ``$-$'' is displayed in case no instances are solved.
To make the table easier to read, columns corresponding to the branching on edges are colored.
In each row, the best result among the configurations is highlighted in bold (a configuration is considered better than another if it solves a greater number of instances, and in case of a tie, if the average elapsed time is lower).

\begin{table}
    \centering
    \footnotesize
    \setlength{\tabcolsep}{3.25pt}
    \begin{tabularx}{\textwidth}{>{\centering\hsize=0.6cm}X>{\centering\hsize=0.6cm}X>{\centering\hsize=0.3cm}X>{\centering\hsize=0.6cm}X>{\columncolor{black!10}}r>{\columncolor{black!10}}rrrr>{\columncolor{black!10}}r>{\columncolor{black!10}}rrrr}
    \toprule
        $p$ & $q$ & $m$ & $\bar k$ & \multicolumn{5}{c}{solved (\%)} & \multicolumn{5}{c}{time (s)} \\
        \cmidrule(lr){5-9} \cmidrule(lr){10-14} 
        ~ & ~ & ~ & ~ & \multicolumn{2}{c}{\textsc{Edg}} & \multicolumn{3}{c}{\textsc{Clr}} & \multicolumn{2}{c}{\textsc{Edg}} & \multicolumn{3}{c}{\textsc{Clr}} \\
        \cmidrule(lr){5-6} \cmidrule(lr){7-9} \cmidrule(lr){10-11} \cmidrule(lr){12-14}
        & & & & \multicolumn{1}{c}{\textsc{Std}} & \multicolumn{1}{c}{\textsc{Alt}} & \textsc{Std} & \textsc{Alt$_1$} & \textsc{Alt$_2$} & \multicolumn{1}{c}{\textsc{Std}} & \multicolumn{1}{c}{\textsc{Alt}} & \textsc{Std} & \textsc{Alt$_1$} & \textsc{Alt$_2$} \\ \midrule
        \multirow{9}{*}{0.25} & \multirow{3}{*}{0.25} & 1 & 15 & 100 & 100 & 100 & 100 & \textbf{100} & 511.2 & 644.8 & 146.6 & 98.1 & \textbf{25.8} \\
        & & 3 & 15 & 100 & 100 & 100 & 100 & \textbf{100} & 331.1 & 657.3 & 145.6 & 104.7 & \textbf{22.4} \\
        & & 5 & 15 & 100 & 100 & 100 & 100 & \textbf{100} & 321.7 & 655.8 & 145.2 & 104.7 & \textbf{22.4} \\
        \cmidrule{2-14}
        & \multirow{3}{*}{0.5} & 1 & 30 & 100 & 100 & 100 & 100 & \textbf{100} & 171.3 & 129.1 & 1853.1 & 294.2 & \textbf{115.0} \\
        & & 3 & 30 & \textbf{100} & 100 & 100 & 100 & 100 & \textbf{91.8} & 130.3 & 2377.1 & 285.0 & 148.8 \\
        & & 5 & 30 & \textbf{100} & 100 & 100 & 100 & 100 & \textbf{91.7} & 130.5 & 2390.1 & 284.7 & 148.7 \\
        \cmidrule{2-14}
        & \multirow{3}{*}{0.75} & 1 & 45 & 100 & \textbf{100} & 80 & 60 & 60 & 155.6 & \textbf{71.0} & 3819.3 & 1668.6 & 1875.9 \\
        & & 3 & 45 & \textbf{100} & 100 & 80 & 60 & 60 & \textbf{41.4} & 71.3 & 3461.0 & 509.2 & 1213.4 \\
        & & 5 & 45 & \textbf{100} & 100 & 80 & 60 & 60 & \textbf{41.4} & 71.3 & 3900.3 & 508.2 & 1213.0 \\
        \midrule
        \multirow{9}{*}{0.5} & \multirow{3}{*}{0.25} & 1 & 15 & 100 & 100 & 100 & \textbf{100} & 100 & 20.9 & 44.9 & 36.1 & \textbf{15.3} & 19.7 \\
        & & 3 & 15 & 100 & 100 & 100 & \textbf{100} & 100 & 35.7 & 18.7 & 25.8 & \textbf{12.1} & 19.7 \\
        & & 5 & 15 & 100 & 100 & 100 & \textbf{100} & 100 & 35.8 & 18.8 & 25.7 & \textbf{12.2} & 19.8 \\
        \cmidrule{2-14}
        & \multirow{3}{*}{0.5} & 1 & 30 & 100 & \textbf{100} & 100 & 100 & 100 & 115.0 & \textbf{84.5} & 1322.1 & 875.7 & 939.9 \\
        & & 3 & 30 & \textbf{100} & 100 & 100 & 100 & 100 & \textbf{39.9} & 44.8 & 357.9 & 664.9 & 1176.8 \\
        & & 5 & 30 & \textbf{100} & 100 & 100 & 100 & 100 & \textbf{39.9} & 44.7 & 358.8 & 665.0 & 1176.0 \\
        \cmidrule{2-14}
        & \multirow{3}{*}{0.75} & 1 & 45 & 100 & \textbf{100} & 0 & 40 & 0 & 96.4 & \textbf{47.7} & -- & 580.7 & -- \\
        & & 3 & 45 & \textbf{100} & 100 & 40 & 20 & 0 & \textbf{44.6} & 52.5 & 2468.1 & 406.2 & -- \\
        & & 5 & 45 & \textbf{100} & 100 & 40 & 20 & 0 & \textbf{44.5} & 52.3 & 2417.8 & 410.7 & -- \\
        \midrule
        \multirow{9}{*}{0.75} & \multirow{3}{*}{0.25} & 1 & 15 & \textbf{100} & 100 & 100 & 100 & 100 & \textbf{1.4} & 3.0 & 5.7 & 6.2 & 2.5 \\
        & & 3 & 15 & 100 & \textbf{100} & 100 & 100 & 100 & 1.7 & \textbf{1.0} & 3.9 & 2.1 & 1.5 \\
        & & 5 & 15 & 100 & \textbf{100} & 100 & 100 & 100 & 1.7 & \textbf{1.0} & 4.0 & 2.1 & 1.5 \\
        \cmidrule{2-14}
        & \multirow{3}{*}{0.5} & 1 & 30 & \textbf{100} & 100 & 100 & 80 & 80 & \textbf{4.1} & 5.3 & 535.6 & 1470.7 & 1753.4 \\
        & & 3 & 30 & \textbf{100} & 100 & 100 & 60 & 100 & \textbf{2.5} & 3.1 & 1107.6 & 532.2 & 590.0 \\
        & & 5 & 30 & \textbf{100} & 100 & 100 & 60 & 100 & \textbf{2.5} & 3.1 & 1144.4 & 533.9 & 870.8 \\
        \cmidrule{2-14}
        & \multirow{3}{*}{0.75} & 1 & 45 & 100 & \textbf{100} & 40 & 40 & 20 & 5.7 & \textbf{5.6} & 39.3 & 10.4 & 0.6 \\
        & & 3 & 45 & 100 & \textbf{100} & 40 & 20 & 20 & 4.8 & \textbf{4.1} & 328.1 & 0.4 & 1.1 \\
        & & 5 & 45 & 100 & \textbf{100} & 60 & 20 & 20 & 5.9 & \textbf{4.1} & 140.2 & 0.4 & 1.1 \\
        \midrule
        \multicolumn{4}{c}{All} & \textbf{100} & 100 & 84 & 76 & 75 & \textbf{83.7} & 111.1 & 1022.8 & 355.6 & 459.7\\
        \bottomrule
    \end{tabularx}
    \caption{Comparison of branching strategies and variable selection rules on instances from Set 1.}
    \label{Tab:set1}
\end{table}

With respect to the branching on edges, all the instances are solved, and the average execution time of \textsc{Alt} is 33\% higher than \textsc{Std}, 83 vs. 111 seconds respectively.
The results suggest that, for a given value of $q$, the average execution time decreases as the graph is denser.
Besides, for given values of $p$ and $q$, the average execution time generally stays the same or decreases as $\bar k$ increases, except for \textsc{Std} on instances with $p = 0.5$ and $q = 0.25$.
Furthermore, the performance is almost the same for $m=3$ and $m=5$.

Regarding the branching on colors, $\textsc{Std}$ solves the highest number of instances, 86\% of the total, but reports the worst average execution time.
In contrast, $\textsc{Alt}_1$ and $\textsc{Alt}_2$ solve fewer instances, 76\% and 75\% respectively, but the average execution times are considerable better, 355 and 459 vs. 1022 seconds respectively.
The results suggest that, for a given value of $p$, the number of solved instances and the average execution time worsen as $\bar k$ increases.
Besides, for equal values of $p$ and $q$, the effects of multiplicity seem to be more notorious than with the branching on edges; but also less predictable, since the performance can get much better or worse as multiplicity increases.

When comparing both strategies, for low $\bar k$ (i.e. $\bar k = 15$), \textsc{Clr} performs far better than \textsc{Edg} on low-density graphs, $\textsc{Clr-Alt}_1$ and $\textsc{Clr-Alt}_2$ obtained slightly better execution times than \textsc{Edg-Std} on medium-density graphs, and all the strategies achieved very low execution times on high-density graphs.
For medium and high values of $\bar k$ (i.e. $\bar k \in \{30,45\}$), the winning strategies are by far the \textsc{Edg}.

These behaviors can be explained by the complete enumeration tree that each type of branch generates, depending on the structure of the instances.
In the branching on edges, given a number of vertices and the color-vertex sets, the tree tends to reach a greater height when the color graphs have lower density, since adding an edge requires more steps to reach a leaf (where each color graph is a complete graph).
In the branching on colors, given a graph and a number of colors, the height of the tree grows with the value of $\bar k$, since forbidding a representative color to a vertex requires more steps to reach a leaf (where each vertex belongs to a single color graph).

Based on these results, a general B\&P algorithm for the WLCP should apply $\textsc{Clr-Alt}_2$ on instances with low-density graphs and low $\bar k$ and \textsc{Edg-Std} otherwise.
In the next computational experiences, we analyze if similar behaviors are observed when these two strategies are evaluated in the other sets of instances.

\subsection{Random instances of $\mu$-coloring}

Table \ref{Tab:set2} shows the results corresponding to the random instances for the $\mu$-coloring problem (Set 2).

\begin{table}
    \centering
    \footnotesize
    \begin{tabularx}{0.86\textwidth}{rrrr>{\columncolor{black!10}}rr>{\columncolor{black!10}}rr}
    \toprule
        $n$ & $p$ & $t$ & $\bar k$ &
        \multicolumn{2}{c}{solved (\%)} ~ & \multicolumn{2}{c}{time (s)} \\ \cmidrule(lr){5-6} \cmidrule(lr){7-8} 
        & & & &  \multicolumn{1}{c}{\textsc{Edg-Std}} & $\textsc{Clr-Alt}_2$ & \multicolumn{1}{c}{\textsc{Edg-Std}} & $\textsc{Clr-Alt}_2$ \\ \midrule
        \multirow{9}{*}{70} & \multirow{3}{*}{0.25} & 0.05 & 13 & \textbf{100} & 0 & 2030.4 & -- \\
        & & 0.1 & 8 & 100 & \textbf{100} & 84.1 & \textbf{22.6} \\
        & & 0.2 & 4 & 100 & \textbf{100} & 307.0 & \textbf{11.6} \\ \cmidrule{2-8}
        & \multirow{3}{*}{0.5} & 0.05 & 13 & \textbf{100} & 33 & 1077.8 & 2734.3 \\
        & & 0.1 & 8 & \textbf{100} & 100 & \textbf{429.3} & 570.5 \\
        & & 0.2 & 5 & 100 & \textbf{100} & 366.5 & \textbf{129.6} \\ \cmidrule{2-8}
        & \multirow{3}{*}{0.75} & 0.05 & 15 & \textbf{100} & 0 & 15.7 & -- \\ 
        & & 0.1 & 9 & \textbf{100} & 100 & \textbf{24.3} & 46.7 \\ 
        & & 0.2 & 5 & 100 & \textbf{100} & 22.8 & \textbf{18.7} \\ \midrule
        \multirow{9}{*}{80} & \multirow{3}{*}{0.25} & 0.05 & 15 & 0 & 0 & -- & -- \\
        & & 0.1 & 8 & 0 & 0 & -- & -- \\
        & & 0.2 & 4 & 0 & \textbf{100} & -- & 780.7 \\ \cmidrule{2-8}
        & \multirow{3}{*}{0.5} & 0.05 & 14 & 0 & 0 & -- & -- \\
        & & 0.1 & 10 & 33 & \textbf{33} & 4239.0 & \textbf{1610.2} \\
        & & 0.2 & 4 & 67 & \textbf{100} & 3075.3 & 335.0 \\ \cmidrule{2-8}
        & \multirow{3}{*}{0.75} & 0.05 & 14 & \textbf{100} & -- & 177.1 & -- \\
        & & 0.1 & 7 & \textbf{100} & 100 & \textbf{30.2} & 678.1 \\
        & & 0.2 & 4 & \textbf{100} & 100 & \textbf{43.5} & 47.7 \\ \midrule
        \multicolumn{4}{c}{All} & \textbf{72} & 59 & 621 & 383 \\ \bottomrule
    \end{tabularx}
    \caption{Comparison of branching strategies and variable selection rules on instances from Set 2.}
    \label{Tab:set2}
\end{table}

The results reaffirm the analysis done in the previous experience, i.e. the performance of \textsc{Clr} is better for instances with low $\bar k$ (which correspond to higher $t$) and \textsc{Edg} for instances with high-density graphs.

By looking at the overall results, \textsc{Edg-Std} solves a larger number of instances than $\textsc{Clr-Alt}_2$, 72\% vs. 59\% respectively.
Again, when the results are desegregated by the structure of the instances, it is possible to see that $\textsc{Clr-Alt}_2$ is convenient for some cases.

For low values of $\bar k$ (with $t=0.2$), $\textsc{Clr-Alt}_2$ is much more effective than \textsc{Edg-Std} on low and medium-density graphs, while the results are fairly similar for high-density graphs. 
On high-density graphs, \textsc{Edg-Std} performs much better than $\textsc{Clr-Alt}_2$ for medium and high values of $\bar k$ (with $t \in \{0.05,0.1\}$).
Besides, on low and medium-density graphs with $n=70$, \textsc{Edg-Std} is superior than $\textsc{Clr-Alt}_2$ for the greatest values of $\bar k$ (with $t = 0.05$).

\subsection{Random instances of GCP}

The results for the random instances of GCP (Set 3) are shown in Table \ref{Tab:set3}. 

\begin{table}[h]
    \centering
    \footnotesize
    \begin{tabularx}{0.67\textwidth}{rr>{\columncolor{black!10}}rr>{\columncolor{black!10}}rr}
        \toprule
        $n$ & $p$ & \multicolumn{2}{c}{solved (\%)} & \multicolumn{2}{c}{time (s)} \\ \cmidrule(lr){3-4} \cmidrule(lr){5-6}
        & & \multicolumn{1}{c}{\textsc{Edg-Std}} & $\textsc{Clr-Alt}_2$ & \multicolumn{1}{c}{\textsc{Edg-Std}} & $\textsc{Clr-Alt}_2$ \\ \midrule
        \multirow{5}{*}{70} & 0.1 & 100 & \textbf{100} & 4.4 & \textbf{0.7} \\
        & 0.3 & 100 & \textbf{100} & 1792.7 & \textbf{81.0} \\
        & 0.5 & 100 & \textbf{100} & 459.8 & \textbf{10.3} \\
        & 0.7 & 100 & \textbf{100} & 21.3 & \textbf{1.8} \\
        & 0.9 & 100 & \textbf{100} & 0.05 & \textbf{0.04} \\ \midrule
        \multirow{5}{*}{80} & 0.1 & 100 & \textbf{100} & 92.6 & \textbf{11.1} \\
        & 0.3 & 0 & \textbf{100} & -- & 687.3 \\
        & 0.5 & 80 & \textbf{100} & 3813.1 & 93.5 \\
        & 0.7 & 100 & \textbf{100} & 219.2 & \textbf{13.8} \\
        & 0.9 & \textbf{100} & 100 & \textbf{0.07} & 0.08 \\ \midrule
        \multicolumn{2}{c}{All} & 88 & \textbf{100} & 641 & 90 \\ \bottomrule
    \end{tabularx}
    \caption{Comparison of branching strategies and variable selection rules on instances from Set 3.}
    \label{Tab:set3}
\end{table}

Recall that instances of GCP initially have $k(v) = 1$ for all vertex $v$, and therefore $\textsc{Clr-Alt}_2$ is expected to perform much better than \textsc{Edg-Std} on low and medium-density graphs.
Once again, the results support this claim. $\textsc{Clr-Alt}_2$ vastly outperforms \textsc{Edg-Std} on graphs with density $p \in \{0.1,0.3,0.5,0.7\}$.
For very high-density graphs (with $p = 0.9$), both strategies perform equally well.

Also, the performance of \textsc{Edg-Std} improves with the increment of $p$, except when we compare $p=0.1$ with $p=0.3$.
Despite the complete enumeration tree is larger for very low-density graphs (in this case, $p=0.1$), it can be observed that BP-LCol quickly finds high-quality feasible solutions to prune, resulting in a very low number of processed nodes.


\subsection{Benchmark instances of GCP}

We compare the results obtained by  BP-LCol with the best B\&P algorithm for GCP from the literature, namely the MMT-BP algorithm by \citet{toth2011}, for the benchmark instances of GCP (Set 4).
MMT-BP combines a B\&P scheme and an effective metaheuristic, called MMT, to obtain an initial feasible solution and initialize the column generation process. 
Also, a Tabu Search algorithm is applied to speed up the pricing routine, and when an entering column cannot be found, an IP model for the MWSSP is directly tackled with CPLEX 10.1.
Two branching strategies are considered: variable branching (\textsc{Var}) and edge branching (\textsc{Edg}).
The first branching strategy is the standard one on the variables of the IP model, which is not robust, while the second is the one originally used by Mehrotra and Trick in \cite{trick1996}.
The variable having the largest fractional part is always selected to branch.

As MMT generates a high-quality initial coloring, we provide BC-LCol with the same initial solution in order to make a fair comparison.
Unfortunately, the adaptation of the solution provided by MMT to our initialization routine is not straightforward and is left as future work.

\begin{table}
    \centering
    \footnotesize
    \setlength{\tabcolsep}{6pt}
    \begin{tabularx}{\textwidth}{lrr>{\columncolor{black!10}}r>{\columncolor{black!10}}rrr>{\columncolor{black!10}}r>{\columncolor{black!10}}rrr}
    \toprule
    \multirow{1}{*}{Name} & \multirow{1}{*}{$n$} & \multirow{1}{*}{$\chi$} & \multicolumn{4}{c}{MMT-BP \cite{toth2011}} & \multicolumn{4}{c}{BP-LCol} \\ \cmidrule(lr){4-7} \cmidrule(lr){8-11}
        ~ & ~ & ~ & \multicolumn{2}{c}{\textsc{Var}} & \multicolumn{2}{c}{\textsc{Edg}} & \multicolumn{2}{c}{\textsc{Edg-Std}} & \multicolumn{2}{c}{$\textsc{Clr-Alt}_2$} \\ \cmidrule(lr){4-5} \cmidrule(lr){6-7} \cmidrule(lr){8-9} \cmidrule(lr){10-11}
        & & & \multicolumn{1}{c}{gap} & \multicolumn{1}{c}{time} & \multicolumn{1}{c}{gap} & \multicolumn{1}{c}{time} & \multicolumn{1}{c}{gap} & \multicolumn{1}{c}{time} & \multicolumn{1}{c}{gap} & \multicolumn{1}{c}{time} \\ 
        & & & \multicolumn{1}{c}{(\%)} & \multicolumn{1}{c}{(s)} & \multicolumn{1}{c}{(\%)} & \multicolumn{1}{c}{(s)} & \multicolumn{1}{c}{(\%)} & \multicolumn{1}{c}{(s)} & \multicolumn{1}{c}{(\%)} & \multicolumn{1}{c}{(s)} \\ \midrule
        \texttt{1-FullIns\_4} & 93 & 5 & 20 & \textit{tl} & $>0$ & \textit{tl} & 0 & 29.7 & \textbf{0} & \textbf{12} \\
        \texttt{1-FullIns\_5} & 282 & 6 & 33.3 & \textit{tl} & $>0$ & \textit{tl} & 33.3 & \textit{tl} & 33.3 & \textit{tl} \\
        \texttt{2-FullIns\_4} & 212 & 6 & 16.7 & \textit{tl} & $>0$ & \textit{tl} & 16.7 & \textit{tl} & 16.7 & \textit{tl} \\
        \texttt{2-FullIns\_5} & 852 & 7 & 28.6 & \textit{tl} & $>0$ & \textit{tl} & 28.6 & \textit{tl} & 28.6 & \textit{tl} \\
        \texttt{3-FullIns\_4} & 405 & 7 & 14.3 & \textit{tl} & $>0$ & \textit{tl} & 14.3 & \textit{tl} & 14.3 & \textit{tl} \\
        \texttt{3-FullIns\_5} & 2030 & 8 & 25 & \textit{tl} & $>0$ & \textit{tl} & 25 & \textit{tl} & 25 & \textit{tl} \\
        \texttt{4-FullIns\_4} & 690 & 8 & 12.5 & \textit{tl} & $>0$ & \textit{tl} & 12.5 & \textit{tl} & 12.5 & \textit{tl} \\
        \texttt{4-FullIns\_5} & 4146 & 9 & \textbf{22.2} & \textit{tl} & $>0$ & \textit{tl} & - & \textit{tl} & - & \textit{tl} \\
        \texttt{5-FullIns\_4}  & 1085 & 9 & 11.1 & \textit{tl} & $>0$ & \textit{tl} & 11.1 & \textit{tl} & 11.1 & \textit{tl} \\
        \texttt{1-Insertions\_4} & 67 & 5 & 40 & \textit{tl} & $>0$ & \textit{tl} & 40 & \textit{tl} & 40 & \textit{tl} \\
        \texttt{1-Insertions\_5} & 202 & 6 & 50 & \textit{tl} & $>0$ & \textit{tl} & 50 & \textit{tl} & 50 & \textit{tl} \\
        \texttt{2-Insertions\_4} & 149 & 5 & 40 & \textit{tl} & $>0$ & \textit{tl} & 40 & \textit{tl} & 40 & \textit{tl} \\
        \texttt{3-Insertions\_3} & 56 & 4 & 25 & \textit{tl} & $>0$ & \textit{tl} & 25 & \textit{tl} & \textbf{0} & 551 \\
        \texttt{3-Insertions\_4} & 281 & 5 & 40 & \textit{tl} & $>0$ & \textit{tl} & 40 & \textit{tl} & 40 & \textit{tl} \\
        \texttt{4-Insertions\_3} & 79 & 4 & 25 & \textit{tl} & $>0$ & \textit{tl} & 25 & \textit{tl} & 25 & \textit{tl} \\
        \texttt{abb313GPIA} & 662 & 9 & \textbf{11.1} & \textit{tl} & $>0$ & \textit{tl} & - & \textit{tl} & - & \textit{tl} \\
        \texttt{DSJC125.5} & 125 & 17 & 0 & 7220 & $>0$ & \textit{tl} & \textbf{0} & \textbf{735} & 0 & 748 \\
        \texttt{DSJC125.9} & 125 & 44 & $>0$ & \textit{tl} & 0 & 1558 & 0 & 57.5 & \textbf{0} & \textbf{55.8} \\
        \texttt{DSJC250.9}  & 250 & 72 & 1.4 & \textit{tl} & $>0$ & \textit{tl} & 1.4 & \textit{tl} & 1.4 & \textit{tl} \\
        \texttt{DSJC500.9}  & 500 & ? & 3.1 & \textit{tl} & $>0$ & \textit{tl} & 3.1 & \textit{tl} & 3.1 & \textit{tl} \\
        \texttt{latin\_square\_10} & 900 & ? & 11.8 & \textit{tl} & $>0$ & \textit{tl} & 11.8 & \textit{tl} & 11.8 & \textit{tl} \\
        \texttt{myciel3} & 11 & 4 & 0 & 1.5 & 0 & 1.5 & 0 & 1.5 & 0 & 1.5 \\
        \texttt{myciel4} & 23 & 5 & 0 & 48.8 & 0 & 3.2 & 0 & 1.9 & \textbf{0} & \textbf{1.7} \\
        \texttt{myciel5} & 47 & 6 & $>0$ & \textit{tl} & $>0$ & tl & 0 & 5746 & \textbf{0} & \textbf{199} \\
        \texttt{myciel6} & 95 & 7 & 42.9 & \textit{tl} & $>0$ & \textit{tl} & 42.9 & \textit{tl} & 42.9 & \textit{tl} \\
        \texttt{myciel7} & 191 & 8 & 37.5 & \textit{tl} & $>0$ & \textit{tl} & 37.5 & \textit{tl} & 37.5 & \textit{tl} \\
        \texttt{queen9\_9} & 81 & 10 & 0 & 13.4 & 0 & 23.8 & 0 & 16.8 & \textbf{0} & \textbf{6.9} \\
        \texttt{queen10\_10} & 100 & 11 & \textbf{0} & \textbf{273} & 0 & 378 & 0 & 1349 & 0 & 668 \\
        \texttt{queen11\_11} & 121 & 11 & \textbf{0} & \textbf{746} & $>0$ & \textit{tl} & 8.3 & \textit{tl} & 0 & 5115 \\
        \texttt{queen12\_12} & 144 & 12 & 7.7 & \textit{tl} & $>0$ & \textit{tl} & 7.7 & \textit{tl} & 7.7 & \textit{tl} \\
        \texttt{queen13\_13} & 169 & 13 & 7.1 & \textit{tl} & $>0$ & \textit{tl} & 7.1 & \textit{tl} & 7.1 & \textit{tl} \\
        \texttt{queen14\_14} & 196 & 14 & 6.7 & \textit{tl} & $>0$ & \textit{tl} & 6.7 & \textit{tl} & 6.7 & \textit{tl} \\
        \texttt{queen15\_15} & 225 & 15 & 6.3 & \textit{tl} & $>0$ & \textit{tl} & 6.3 & \textit{tl} & 6.3 & \textit{tl} \\
        \texttt{queen16\_16} & 256 & 16 & \textbf{5.9} & \textit{tl} & $>0$ & \textit{tl} & - & \textit{tl} & - & \textit{tl} \\ \bottomrule
    \end{tabularx}
    \caption{Comparison with the best B\&P algorithm from the literature for GCP on instances from Set 4.}
    \label{Tab:set4}
\end{table}

The first 3 columns of Table \ref{Tab:set4} report the name, the number of vertices $n$, and the chromatic number $\chi$ of each instance, and the rest of the columns report the results of both algorithms.
For each branching strategy, we report the gap (between the best lower and upper bounds achieved within the time limit) and the execution time in seconds.
As MMT-BP was executed on a Pentium IV 2.4 GHz with 1 GB RAM, the execution times reported in columns 5 and 7 have been adjusted in order to make a comparison with similar computing times with our machine, which we estimate to be 2.5 times faster (they reported 7 s to solve the benchmark instance r500.5.b).
The time limit for MMT-BP was 36000 s, including the time spent by MMT to find an initial coloring, which corresponds to a time limit of 14400 s in BC-LCol. To make the table easier to read, columns corresponding to \textsc{Var} and \textsc{Edg-Std} are colored.
In each row, the best result among the algorithms is highlighted in bold.
A ``\textit{tl}'' entry means that the time limit was reached, a ``-'' entry means that the gap cannot be computed since the resolution of the initial linear relaxation exceeded the time limit, and a ``$>0$'' entry in the columns of MMT-BP means that there is a positive gap but the exact value was not reported in \cite{toth2011}.

MMT-BP proves optimality in 6 instances using \textsc{Var} and 5 instances using \textsc{Edg}, with \textsc{Var} being more convenient than \textsc{Edg} except for \texttt{DSJC125.9} and \texttt{myciel4}.
BC-LCol proves optimality in 8 instances using \textsc{Edg-Std} and 10 instances using $\textsc{Clr-Alt}_2$, with $\textsc{Clr-Alt}_2$ being more convenient than \textsc{Edg-Std} except for \texttt{DSJC125.9} (although the difference here is subtle: 57.5 s versus 55.8 s).

In particular, BP-LCol proves optimality for 3 instances that cannot be solved by MMT-BP within the time limit (\texttt{1-FullIns\_4}, \texttt{3-Insertions\_3}, and \texttt{myciel5}).
Considering the instances solved by both algorithms, ours reports better execution times in 4 instances (\texttt{DSJC125.5}, \texttt{DSJC125.9}, \texttt{myciel4}, and \texttt{queen9\_9}) and worse times in 2 instances (\texttt{queen10\_10} and \texttt{queen11\_11}).
In order to explain this behavior, recall that, for instances coming from GCP, the algorithms only differ in the preprocessing and the resolution of the linear relaxations (initialization and pricing routine).
 Then, we executed BC-LCol again without preprocessing the instances, but similar results were obtained, indicating that the differences must come from the resolution of the linear relaxations.
Another point that supports this claim is observed in the time required to solve the initial linear relaxations, where BC-LCol is almost always several times faster, except for \texttt{4-FullIns\_5}, \texttt{abb313GPIA}, and \texttt{queenN\_N} with $\texttt{N} \in \{10,11,12,13,14,15,16\}$.

It is worth mentioning that BC-LCol run out of time when solving the initial linear relaxation of \texttt{4-FullIns\_5}, \texttt{abb313GPIA}, and \texttt{queen16\_16}. Specifically, the resolution of the MWSSP exceeded the time limit (4 hours) looking for a stable set in a particular color graph.
In these cases, the pricing routine could be modified to abort the current color after a certain time and continue with the next one.
If none entering column could be found, then the aborted colors (if any) must be fully scanned.

Similar gaps were reported by both algorithms for the remaining instances.

These results show that BC-LCol is competitive with others from the literature for solving GCP.
Also, the enumeration routine developed by \cite{held2011} seems to be more efficient in solving the MWSSP compared to the approach followed by \cite{toth2011}, except for the aforementioned instances.
This suggests that the approach presented in this work, i.e. treating GCP instances as LCP instances and dividing the solution space according to branching on colors might be better than applying MMT-BP with \texttt{Var}.
To known for sure, MMT-BP should be re-implemented to solve the MWSSP with the enumeration routine of \cite{held2011}.

\subsection{Particular instances of WLCP}

Finally, the results for particular instances of WLCP (Set 5) are discussed.
These instances have from 200 to 900 vertices, the color graphs have a low to medium average density, and the value of the parameter $\bar{k}$ is generally high (except for the instances of Precoloring Extension Problem and a couple of instances of MWSCP).
Although BP-LCol is not expected to perform efficiently here, these instances are still interesting to confirm the behavior of the branching strategies on these particular structures.

The 3 instances for the Precoloring Extension Problem have low-density color graphs and few colors, thus the branching on colors is expected to perform better than the branching on edges.
Besides, the chromatic number is equal to the number of colors, which means that there are few feasible solutions and all of them are optimal.
Both strategies solve \texttt{order18holes120} in less than 1 s and only $\textsc{Clr-Alt}_2$ solves \texttt{order30holes320} in 408 s.
No solution was found within the time limit for \texttt{order30holes316}.

The 10 instances for the MWSCP have edgeless color graphs, which are expected to be the most challenging structure for the branching on edges, and the parameter $\bar{k}$ varies over a wide range of values, where the higher values are expected to be the most challenging for the branching on colors.
\textsc{Edg-Std} solves 2 instances (\texttt{scp41} in less than 1 s and \texttt{rail516} in 123 s), and $\textsc{Clr-Alt}_2$ solves those two instances with similar execution times and \texttt{scpc1} in 611 s.
\textsc{Edg-Std} reports positive gaps in 3 instances, \texttt{scpcyc06} (36\%), \texttt{scpa1} (8\%), and \texttt{scpc1} (3\%), whereas $\textsc{Clr-Alt}_2$ reports positive gaps in 4 instances, \texttt{scpcyc06} (41\%), \texttt{scpa1} (11\%), \texttt{scpclr1} (34\%) and \texttt{scpcyc07} (52\%).
This means that $\textsc{Clr-Alt}_2$ found feasible solutions for all the instances with $\bar{k} \leq 80$.
In the other cases, BP-LCol ran out of memory before the time limit, since no feasible solutions were found to prune the tree.
As mentioned, these results are not intended to compete with those achieved by specific algorithms for the MWSCP.

The 3 instances from a transport application have medium-density color graphs and large lists of colors, which is apparently the most complex structure for both branching strategies.
The instance \texttt{trips215drivers60} is solved in 111 s with \textsc{Edg-Alt} vs. 759 s with $\textsc{Clr-Alt}_2$.
Again, for the remaining instances, the B\&P algorithm ran out of memory before finding a feasible solution.

\section{Conclusions and future work}\label{Sec:Conclusions}

In this work, we present a B\&P algorithm (BP-LCol) to solve the weighted version of the List Coloring Problem, which has many applications and generalizes several problems. 
As far as we are concerned, this is the first exact algorithm based on ILP for this problem.

At the beginning of this work, the focus is on the computational complexity of WLCP restricted to certain families of instances, as a way of exposing the importance of the structure of the input, which depends not only on the graph but also on the color graphs, i.e. the subgraphs induced by the vertices that have a given color in its list.
These considerations anticipate that designing a general algorithm,  efficient for the wide variety of possible instances, would not be easy.

The introduction of the concept of indistinguishable colors contributes to this research in two ways.
On the one hand, 
it allows us to extend to WLCP the well-known procedure of initially coloring a maximal clique on GCP.
On the other hand, the number of variables
and  constraints in the proposed ILP formulation can be reduced, resulting in the same as proposed by \citet{trick1996} when restricted to GCP instances.

The incorporation of dummy variables helps to find a set of columns to initialize the resolution of the linear relaxation by column generation.
The existence of a theoretical weight $M$ for these new variables is proved, such that an optimal solution of the extended linear relaxation can straightforwardly derive an optimal solution of the original linear relaxation or tells that is infeasible, avoiding a classical two-phase initialization.
Two branching strategies were presented: branching on edges adapts the branching used originally by \citet{trick1996} for GCP, and branching on colors is inspired in the division space of DSATUR and the B\&C algorithm developed by \citet{mendezdiaz2006} for GCP.


A large number of computational experiments  is used to evaluate the performance of BP-LCol, considering instances with a wide variety of structures.
The results on random instances of WLCP confirm that the algorithm behaves differently depending on these structures.
Branching on edges improves as the color graphs are denser, whereas branching on colors improves as lists have fewer representative colors.
The results for other sets of instances experimentally reconfirm this behavior.
In particular, the results for random and benchmark instances of GCP suggest that treating them as LCP instances could be beneficial. Indeed, it is possible to take advantage of the developments for this more general problem, such as the branching on colors and the preprocessing.

\medskip

From this research, some questions remain open for future work.
For example, the impact of the weights of the colors on the algorithm efficiency, as well as the existence of other structures in the instances that influence the performance. 

A further study of the column generation process is another possible line of research.
Currently, at most one entering column is searched for each color before re-optimizing the linear relaxation.
Some alternatives to this approach are, for example, to allow as many entering columns for each color as its multiplicity, re-optimize the linear relaxation when a given number of entering columns have been found, and in that case, sort the colors by some criteria, set a time limit for the resolution of the MWSSP, take advantage of colors that have the same color-vertex sets but different weights (to avoid solving twice the MWSSP), maintain a pool of entering columns, among others.


The proposed weight $M = 1000$ for dummy variables during initialization works well in computational experiences.
To cover the hypothetical case in which it fails, the implementation of the two-phase method 
should be added to BP-LCol.

Finally, we believe that the incorporation of a (meta)heuristic for the WLCP would be highly beneficial.
An initial feasible solution, i.e. an upper bound, would allow node pruning from the beginning and would help to avoid the memory problems reported in computational experiences on large instances. 
Additionally, the heuristic can be used to generate an appropriate initial set of columns to initialize the column generation procedure.
Even if the heuristic only finds a partial list coloring, the initialization can still use this solution by adding dummy variables for the uncolored vertices.

\section*{References}

\bibliographystyle{apalike} \bibliography{references}

\end{document}